\documentclass[conference]{IEEEtran}
\IEEEoverridecommandlockouts
\usepackage{cite}
\usepackage{amsmath,amssymb,amsfonts,amsthm}
\usepackage{algorithmic}
\usepackage{graphicx}
\usepackage{textcomp}
\usepackage{xcolor}
\usepackage{pifont}
\usepackage{relsize}
\usepackage{multirow}
\usepackage{url}
\usepackage{mathrsfs}
\usepackage{color}
\usepackage{algorithmic}
\usepackage{xspace}
\usepackage{float}
\usepackage{multirow}
\usepackage{contour}
\usepackage{booktabs,threeparttable}
\usepackage{hyperref}
\usepackage[ruled, vlined, linesnumbered]{algorithm2e}
\usepackage{enumitem}
\usepackage{balance}
\usepackage{caption}
\usepackage{subcaption}
\usepackage{fancyhdr}
\newtheorem{definition}{Definition}

\newtheorem{theorem}{Theorem}
\newtheorem{lemma}{Lemma}

\newtheorem{corollary}{Corollary}

\newcommand{\HM}{\texttt{HM}}
\newcommand{\Var}{\text{\it Var}}

\def\header{\vspace{2.5mm} \noindent}

\newcommand{\Lap}{\text{\it Lap}}
\newcommand{\E}{\mathbb{E}}

\newcommand{\MaxVar}{\text{\it MaxVar}}
\newcommand{\pdf}{\text{\it pdf}}
\newcommand{\p}{\prime}

\newcommand{\eat}[1]{}

\begin{document}

\title{Collecting and Analyzing Multidimensional Data with Local Differential Privacy}

\author{Ning Wang$^{1}$, Xiaokui Xiao$^2$, Yin Yang$^3$, Jun Zhao$^4$, Siu Cheung Hui$^4$, Hyejin Shin$^5$, Junbum Shin$^5$, Ge Yu$^{6}$
\vspace{1.8mm}
\\
\fontsize{10}{10}\selectfont\itshape
$^1$School of Information Science and Engineering, Ocean University of China\\
$^2$School of Computing, National University of Singapore, Singapore\\
$^3$Division of Information and Computing Techologies, College of Science and Engineering, Hamad Bin Khalifa University, Qatar\\
$^4$School of Computer Science and Engineering, Nanyang Technological University, Singapore\\
$^5$Samsung Research, Samsung Electronics, Korea\\
$^6$School of Computer Science and Engineering, Northeastern University, China\vspace{1.8mm}
\\
\fontsize{9}{9}\selectfont\ttfamily\upshape
$^1$wangning8687@ouc.edu.cn, $^2$xkxiao@nus.edu.sg, $^3$yyang@hbku.edu.qa, $^4$\{junzhao,\,asschui\}@ntu.edu.sg, \\ $^5$\{hyejin1.shin,\,junbum.shin\}@samsung.com, $^6$yuge@mail.neu.edu.cn
}

\maketitle

\thispagestyle{fancy} \pagestyle{fancy}

\begin{abstract}

Local differential privacy (LDP) is a recently proposed privacy standard for collecting and analyzing data, which has been used, e.g., in the Chrome browser, iOS and macOS. In LDP, each user perturbs her information locally, and only sends the randomized version to an aggregator who performs analyses, which protects both the users and the aggregator against private information leaks. Although LDP has attracted much research attention in recent years, the majority of existing work focuses on applying LDP to complex data and/or analysis tasks. In this paper, we point out that the fundamental problem of collecting multidimensional data under LDP has not been addressed sufficiently, and there remains much room for improvement even for basic tasks such as computing the mean value over a single numeric attribute under LDP. Motivated by this, we first propose novel LDP mechanisms for collecting a numeric attribute, whose accuracy is at least no worse (and usually better) than existing solutions in terms of worst-case noise variance. Then, we extend these mechanisms to multidimensional data that can contain both numeric and categorical attributes, where our mechanisms always outperform existing solutions regarding worst-case noise variance. As a case study, we apply our solutions to build an \mbox{LDP-compliant} stochastic gradient descent algorithm (SGD), which powers many important machine learning tasks. Experiments using real datasets confirm the effectiveness of our methods, and their advantages over existing solutions.

\end{abstract}



\section{Introduction} \label{sec:intro}

Local differential privacy (LDP), which has been used in well-known systems such as Google Chrome~\cite{RAPPOR2014}, Apple iOS and macOS~\cite{tang2017privacy}, and Microsoft Windows Insiders~\cite{ding2017collecting}, is a rigorous privacy protection scheme for collecting and analyzing sensitive data from individual users. Specifically, in LDP, each user perturbs her data record locally to satisfy differential privacy \cite{DworkLaplace}, and sends only the randomized, differentially private version of the record to an {\em aggregator}. The latter then performs computations on the collected noisy data to estimate statistical analysis results on the original data. For instance, in \cite{RAPPOR2014}, Google as an aggregator collects perturbed usage information from users of the Chrome browser, and estimates, e.g., the proportion of users running a particular operating system. Compared with traditional privacy standards such as differential privacy in the centralized setting \cite{DworkLaplace}, which typically assume a trusted data curator who possesses a set of sensitive records, LDP provides a stronger privacy assurance to users, as the true values of private records never leave their local devices. Meanwhile, LDP also protects the aggregator against potential leakage of users' private information (which happened to AOL\footnote{\url{https://en.wikipedia.org/wiki/AOL_search_data_leak}} and Netflix\footnote{\url{https://www.wired.com/2009/12/netflix-privacy-lawsuit/}} with serious consequences), since the aggregator never collects exact private information in the first place. In addition, LDP satisfies the strong and rigorous privacy guarantees of differential privacy; i.e., the adversary (which includes the aggregator in LDP) cannot infer sensitive information of an individual with high confidence, regardless of the adversary's background knowledge.

Although LDP has attracted much attention in recent years, the majority of existing solutions focus on applying LDP to complex data types and/or data analysis tasks, as reviewed in Section \ref{sec:related}. Notably, the fundamental problem of collecting numeric data has not been addressed sufficiently. As we explain in Section \ref{sec:numeric_existing}, in order to release a numeric value in the range $[-1, 1]$ under LDP, currently the user has only two options: (i) the classic Laplace mechanism \cite{DworkLaplace}, which injects \emph{unbounded noise} to the exact data value, and (ii) a recent proposal by Duchi~et~al.~\cite{DuchiJW18}, which releases a perturbed value that \emph{always falls outside the original data domain}, i.e., $[-1, 1]$. Further, it is non-trivial to extend these methods to handle multidimensional data. As elaborated in Section \ref{sec:multi}, a straightforward extension of a single-attribute mechanism, using the composition property of differential privacy, leads to suboptimal result accuracy. Meanwhile, the multidimensional version of~\cite{DuchiJW18}, though asymptotically optimal in terms of worst-case error, is complicated and involves a large constant. Finally, to our knowledge, there is no existing solution that can perturb multidimensional data containing both numeric and categorical data with optimal worst-case error.

This paper addresses the above challenges and makes several major contributions. First, we propose two novel mechanisms, namely Piecewise Mechanism (PM) and Hybrid Mechanism (HM), for collecting a single numeric attribute under LDP, which obtain higher result accuracy compared to existing methods. In particular, HM is built upon PM, and has a worse-case noise variance that is at least no worse  (and usually better) than existing solutions. Then, we extend both PM and HM to multidimensional data with both numeric and categorical attributes with an elegant technique that achieves asymptotic optimal error, while remaining conceptually simple and easy to implement. Further, our fine-grained analysis reveals that although both \cite{DuchiJW18} and the proposed methods obtain asymptotically optimal error bound on multidimensional numeric data, the former involves a larger constant than our solutions. Table~\ref{table-comparison} summarizes the main theoretical results in this paper, which are confirmed in our experiments.

\begin{table}[!t]
\centering
\caption{Main theoretical results comparing the proposed mechanisms PM and HM, as well as Duchi~et~al.'s solution~\cite{DuchiJW18}. The terms $\MaxVar_{\text{PM}}$, $\MaxVar_{\text{HM}}$, and $\MaxVar_{\text{Du}}$ denote the worst-case noise variance of these three methods, respectively, for perturbing a $d$-dimensional numeric tuple under $\epsilon$-local differential privacy (elaborated in Section~\ref{sec:prelim}). In addition, $\epsilon^{\#} = \textstyle  \ln \left( \frac{7+4\sqrt{7}+2\sqrt{20+14\sqrt{7}}}{9} \right) \approx 1.29$ and $\epsilon^* = \textstyle  \ln \left(  \frac{-5 + 2\sqrt[3]{6353 - 405 \sqrt{241}} \, + \, 2\sqrt[3]{6353 + 405 \sqrt{241}}}{27} \right)     \approx 0.61$.\vspace{-4pt}}
\label{table-comparison}
\begin{tabular}{|l|l|l|}
\hline
\multicolumn{2}{|l|}{Setting}                                      & Result                          \\ \hline
$d > 1$                  & $\epsilon > 0$                          & $\MaxVar_{\text{HM}}<\MaxVar_{\text{PM}}<\MaxVar_{\text{Du}}$ \\ \hline
\multirow{4}{*}{$d = 1$} & $\epsilon > \epsilon^{\#}$              & $\MaxVar_{\text{HM}}<\MaxVar_{\text{PM}}<\MaxVar_{\text{Du}}$ \\ \cline{2-3} 
                         & $\epsilon = \epsilon^{\#}$              & $\MaxVar_{\text{HM}}<\MaxVar_{\text{PM}}=\MaxVar_{\text{Du}}$ \\ \cline{2-3} 
                         & $\epsilon^{*}<\epsilon < \epsilon^{\#}$ & $\MaxVar_{\text{HM}}<\MaxVar_{\text{Du}}<\MaxVar_{\text{PM}}$ \\ \cline{2-3} 
                         & $0<\epsilon \leq \epsilon^{*}$          & $\MaxVar_{\text{HM}}=\MaxVar_{\text{Du}}<\MaxVar_{\text{PM}}$ \\ \hline 
\end{tabular}
\vspace{-15pt}
\end{table}

As a case study, using the proposed mechanisms as building blocks, we present an \mbox{LDP-compliant} algorithm for stochastic gradient descent (SGD), which can be applied to train a broad class of machine learning models based on empirical risk minimization, e.g., linear regression, logistic regression and SVM classification. Specifically, SGD iteratively updates the model based on gradients of the objective function, which are collected from individuals under LDP. Experiments using several real datasets confirm the high utility of the proposed methods for various types of data analysis tasks.

In the following, Section~\ref{sec:prelim} provides the necessary background on LDP. Sections~\ref{sec:basic} presents the proposed fundamental mechanisms for collecting a single numeric attribute under LDP, while Section~\ref{sec:multi} describes our solution for collecting and analyzing multidimensional data with both numeric and categorical attributes. Section \ref{sec:riskmini} applies our solution to common data analytics tasks based on SGD, including linear regression, logistic regression, and support vector machines (SVM) classification. Section \ref{sec:exp} contains an extensive set of experiments. Section \ref{sec:related} reviews related work. Finally, Section \ref{sec:conclusion} concludes the paper.

\section{Preliminaries} \label{sec:prelim}

In the problem setting, an \emph{aggregator} collects data from a set of \emph{users}, and computes statistical models based on the collected data. The goal is to maximize the accuracy of these statistical models, while preserving the privacy of the users. Following the local differential privacy model \cite{RAPPOR2014, BS15, DuchiJW18}, we assume that the aggregator already knows the identities of the users, but not their private data. Formally, let $n$ be the total number of users, and $u_i$ ($1 \leq i \leq n$) denote the $i$-th user. Each user $u_i$'s private data is represented by a tuple $t_i$, which contains $d$ attributes $A_1, A_2, \ldots, A_d$. These attributes can be either numeric or categorical. Without loss of generality, we assume that each numeric attribute has a domain $[-1, 1]$, and each categorical attribute with $k$ distinct values has a discrete domain $\{1, 2, \ldots, k\}$.

To protect privacy, each user $u_i$ first perturbs her tuple $t_i$ using a randomized \emph{perturbation function} $f$. Then, she sends the perturbed data $f(t_i)$ to the aggregator instead of her true data record $t_i$. 
Given a privacy parameter $\epsilon > 0$ that controls the privacy-utility tradeoff, we require that $f$ satisfies $\epsilon$-\emph{local differential privacy} ({\em $\epsilon$-LDP}) \cite{RAPPOR2014}, defined as follows:
\begin{definition}[$\epsilon$-local differential privacy]\label{def:ldp}
A randomized function $f$ satisfies $\epsilon$-local differential privacy if and only if for any two input tuples $t$ and $t^\prime$ in the domain of $f$, and for any output $t^*$ of $f$, we have:
 \begin{align}
\Pr\big[f(t) = t^*\big] \leq \exp(\epsilon) \cdot \Pr\big[f(t^\prime) = t^*\big]. \label{eqn:def:LDP}
 \end{align}
\end{definition}

The notation $\Pr[\cdot]$ means probability. If $f$'s output is continuous, $\Pr[\cdot]$ in (\ref{eqn:def:LDP}) is replaced by the probability density function. Basically, local differential privacy is a special case of differential privacy \cite{DworkR14} where the random perturbation is performed by the users, not by the aggregator. 
According to the above definition, the aggregator, who receives the perturbed tuple $t^*$, cannot distinguish whether the true tuple is $t$ or another tuple $t'$ with high confidence (controlled by parameter $\epsilon$), regardless of the background information of the aggregator. This provides \textit{plausible deniability} to the user~\cite{cormode2018privacy}. 

We aim to support two types of analytics tasks under \mbox{$\epsilon$-LDP}: (i) mean value and frequency estimation and (ii) machine learning models based on empirical risk minimization. In the former, for each numerical attribute $A_j$, we aim to estimate the mean value of $A_j$ over all $n$ users, $\frac{1}{n} \sum_{i=1}^n {t_i[A_j]}$. For each categorical attribute $A_j'$, we aim to estimate the frequency of each possible value of $A_j'$. Note that value frequencies in a categorical attribute $A_j'$ can be transformed to mean values once we expand $A_j'$ into $k$ binary attributes using one-hot encoding.
Regarding empirical risk minimization, we focus on three common analysis tasks: linear regression, logistic regression, and support vector machines (SVM) \cite{CortesV95}.

Unless otherwise specified, all expectations in this paper are taken over the random choices made by the algorithms considered. We use $\E[\cdot]$ and $\Var[\cdot]$ to denote a random variable's expected value and variance, respectively.


\section{Collecting A Single Numeric Attribute} \label{sec:basic}

This section focuses on the problem of estimating the mean value of a numeric attribute by collecting data from individuals under $\epsilon$-LDP. Section \ref{sec:numeric_existing} reviews two existing methods, Laplace Mechanism \cite{DworkLaplace} and Duchi~et~al.'s solution \cite{DuchiJW18}, and discusses their deficiencies. Then, Section \ref{sec:PM} describes a novel solution, called Piecewise Mechanism (PM), that addresses these deficiencies and usually leads to higher (or at least comparable) accuracy than existing solutions. Section \ref{sec:HM} presents our main proposal, called Hybrid Mechanism (HM), whose worst-case result accuracy is no worse than PM and existing methods, and is often better than all of them.

\subsection{Existing Solutions}\label{sec:numeric_existing}

\header
\textbf{Laplace mechanism and its variants.} A classic mechanism for enforcing differential privacy is the \emph{Laplace Mechanism} \cite{DworkLaplace}, which can be applied to the LDP setting as follows. For simplicity, assume that each user $u_i$'s data record $t_i$ contains a single numeric attribute whose value lies in range $[-1, 1]$. In the following, we abuse the notation by using $t_i$ to denote this attribute value. Then,
we define a randomized function that outputs a perturbed record $ t^*_i = t_i + \Lap\left(\frac{2}{\epsilon}\right),$
 where $\Lap(\lambda)$ denotes a random variable that follows a Laplace distribution of scale $\lambda$, with the following probability density function: $\pdf(x) = \frac{1}{2\lambda} \exp\left(-\frac{|x|}{\lambda}\right).$


Clearly, this estimate $t^{*}_i$ is unbiased, since the injected Laplace noise $\Lap\left(\frac{2}{\epsilon}\right)$ in each $t^*_i$ has zero mean. In addition, the variance in $t^{*}_i$ is $\frac{8}{\epsilon^2}$.
Once the aggregator receives all perturbed tuples, it simply \vspace{1pt} computes their average $\frac{1}{n} \sum_{i = 1}^n t^*_i$ as an estimate of the mean with error scale $O\big(\frac{1}{\epsilon\sqrt{n}}\big)$.

Soria-Comas and Domingo-Ferrer~\cite{Soria-ComasD13} propose a more sophisticated variant of Laplace mechanism, hereafter referred to as \emph{SCDF}, that obtains improved result accuracy for multi-dimensional data. Later, Geng et al.~\cite{GengKOV15} propose \emph{Staircase mechanism}, which achieves optimal performance for unbounded input values (e.g., from a domain of $(-\infty, +\infty)$). Specifically, for a single attribute value $t_i$, both methods inject random noise $n_i$ drawn from the following piece-wise constant probability distribution:
\begin{equation} \label{eqn:sc-scdf}
\small
\pdf(n_i=x) =
\begin{cases}
\frac{a(m)}{e^{j\epsilon}}, & \textrm{if $x \in \left[-m -2(j+1) ,-m-2j\right], j\in \mathbb{N}$,}  \\[5pt]
a(m), & \textrm{if $x \in \left[-m, m\right]$,} \\[5pt]
\frac{a(m)}{e^{j\epsilon}}, & \textrm{if $x \in \left[m +2j ,m+2(j+1)\right], j \in \mathbb{N}$.}
\end{cases}
\end{equation}
In SCDF, $m =2 \cdot \frac{1 -\exp(-\epsilon) - \epsilon\exp(-\epsilon)}{\epsilon -\epsilon\exp(-\epsilon)}$ and $a(m)=\frac{\epsilon}{4}$; in Staircase mechanism, $m =\frac{2}{1 + e^{\epsilon/2}}$ and $a(m) = \frac{1-e^{-\epsilon}}{2m + 4e^{-\epsilon} - 2me^{-\epsilon}}$. Note that the optimality result in~\cite{GengKOV15} does not apply to the case with bounded inputs. We experimentally compare the proposed solutions with both SCDF and Staircase in Section \ref{sec:exp}.

\begin{algorithm}[t]
\caption{Duchi~et~al.'s Solution~\cite{DuchiJW18} for \mbox{One-Dimensional} Numeric Data.}\label{alg:duchi-onedimension}
\SetKwInOut{Input}{input}
\SetKwInOut{Output}{output}
\Input{tuple $t_i \in [-1,1]$ and privacy parameter $\epsilon.$}
\Output{tuple $t^*_i \in \left\{-\frac{e^\epsilon+1}{e^\epsilon-1}, \:\: \frac{e^\epsilon+1}{e^\epsilon-1}\right\}.$}
    Sample a Bernoulli variable $u$ such that
    \parbox{45mm}{$$\Pr[u = 1] = \frac{e^\epsilon-1}{2e^\epsilon + 2} \cdot t_i + \frac{1}{2}$$}\;
    \If{$u = 1$}
    {
        $t^*_i = \frac{e^\epsilon+1}{e^\epsilon-1} $\;
    }
    \Else
    {
        $t^*_i = -\frac{e^\epsilon+1}{e^\epsilon-1} $\;
    }
    \Return $t^*_i$
\end{algorithm}

\header
\textbf{Duchi~et~al.'s solution.} Duchi~et~al.~\cite{DuchiJW18} propose a method to perturb multidimensional numeric tuples under LDP. Algorithm~\ref{alg:duchi-onedimension} illustrates Duchi~et~al.'s solution~\cite{DuchiJW18} for the one-dimensional case. (We discuss the multidimensional case in Section~\ref{sec:multi}.) In particular, given a tuple $t_i \in [-1, 1]$, the algorithm returns a perturbed tuple $t^*_i$ that equals either $\frac{e^\epsilon+1}{e^\epsilon-1}$ or $-\frac{e^\epsilon+1}{e^\epsilon-1}$, with the following probabilities:
\begin{equation} \label{eqn:basic-improved}
\Pr\big[t^*_i = x \mid t_i\big] =
\begin{cases}
\frac{e^\epsilon-1}{2e^\epsilon + 2} \cdot t_i + \frac{1}{2}, & \textrm{if $x = \frac{e^\epsilon+1}{e^\epsilon-1}$,}  \\[5pt]
-\frac{e^\epsilon-1}{2e^\epsilon + 2} \cdot t_i + \frac{1}{2}, & \textrm{if $x = -\frac{e^\epsilon+1}{e^\epsilon-1} $.}
\end{cases}
\end{equation}
Duchi~et~al.\ prove that $t^*_i$ is an unbiased estimator of the input value $t_i$. In addition, the variance of $t^*_i$ is:
\begin{align} \label{eqn:basic-duchi-variance}
&\Var[t^*_i] = \E\left[(t^*_i)^2\right] - (\E[t^*_i])^2 \nonumber \\
&= \textstyle (\frac{e^\epsilon+1}{e^\epsilon-1})^2\cdot\frac{t_i\cdot (e^\epsilon-1) + e^\epsilon + 1}{2e^\epsilon + 2} \hspace{-1pt}+\hspace{-1pt} (-\frac{e^\epsilon+1}{e^\epsilon-1})^2\cdot\frac{-t_i\cdot (e^\epsilon-1) + e^\epsilon + 1}{2e^\epsilon + 2} \hspace{-1pt}-\hspace{-1pt} {t_i}^2 \nonumber\\
&=\textstyle \left(\frac{e^\epsilon+1}{e^\epsilon-1}\right)^2 - {t_i}^2.
\end{align}
Therefore, the worst-case variance of $t^*_i$ equals $\left(\frac{e^\epsilon+1}{e^\epsilon-1}\right)^2$, and it occurs when $t_i = 0$. Upon receiving the perturbed tuples output by Algorithm~\ref{alg:duchi-onedimension}, the aggregator simply computes the average value of the attribute over all users to obtain an estimated mean.

\header
\textbf{Deficiencies of existing solutions.} Fig.~\ref{fig:FMCom} illustrates the worst-case variance of the noisy values returned by the Laplace mechanism and Duchi~et~al.'s solution, when $\epsilon$ varies. Duchi~et~al.'s solution offers considerably smaller variance than the Laplace mechanism when $\epsilon \le 2$, but is significantly outperformed by the latter when $\epsilon$ is large. To explain, recall that Duchi~et~al.'s solution returns either $t^*_i = \frac{e^\epsilon+1}{e^\epsilon-1}$ or $t^*_i = -\frac{e^\epsilon+1}{e^\epsilon-1}$, even when the input tuple $t_i = 0$. As such, the noisy value $t^*_i$ output by Duchi~et~al.'s solution always has an absolute value $\frac{e^\epsilon+1}{e^\epsilon-1} > 1$, due to which $t^*_i$'s variance is always larger than $1$ when $t_i = 0$, regardless of how large the privacy budget $\epsilon$ is. In contrast, the Laplace mechanism incurs a noise variance of $8/\epsilon^2$, which decreases quadratically with
the increase of $\epsilon$, due to which it is preferable when $\epsilon$ is large. However, when $\epsilon$ is small, the relatively ``fat'' tail of the Laplace distribution leads to a large noise variance, whereas Duchi~et~al.'s solution does not suffer from this issue since it confines $t^*_i$ within a relatively small range $\big[-\frac{e^\epsilon+1}{e^\epsilon-1}, \frac{e^\epsilon+1}{e^\epsilon-1}\big]$.
SCDF and Staircase mechanism suffer from the same issue as the Laplace mechanism, as demonstrated in our experiments in Section \ref{sec:exp}.

\begin{figure}[t]
\centering
\includegraphics[width=0.9\columnwidth]{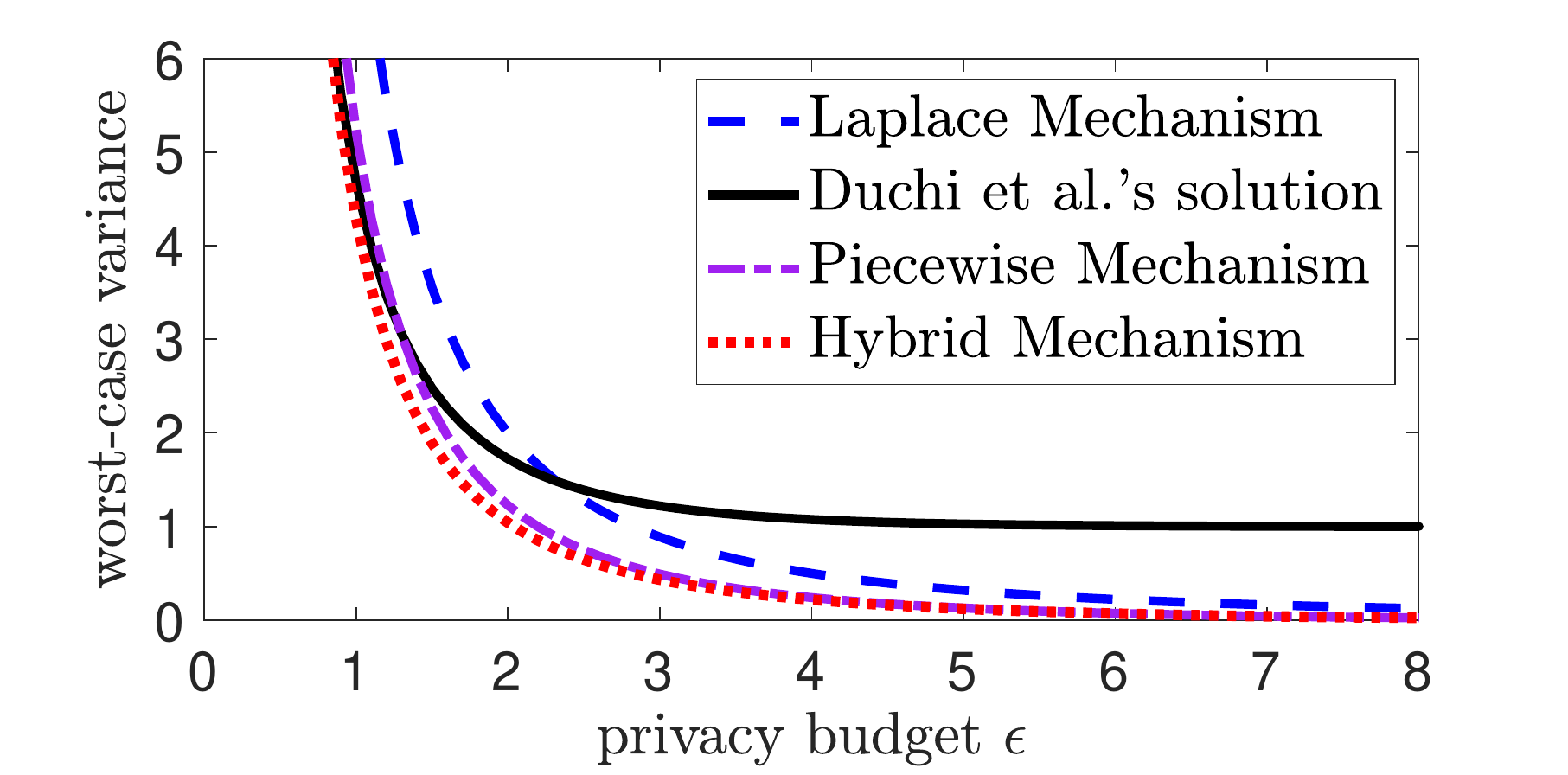}
\vspace{-1mm}
\caption{Different mechanisms' worst-case noise variances for one-dimensional numeric data versus the privacy budget $\epsilon$. Our Piecewise Mechanism and Hybrid Mechanism will be discussed in Sections \ref{sec:PM} and \ref{sec:HM}, respectively.\vspace{-10pt}}
\label{fig:FMCom}
\end{figure}

\eat{
\begin{figure*}[!htbp]

\begin{subfigure}{0.265\textwidth}
\includegraphics[width=1\linewidth]{comparetruevalue0.pdf}
\caption{$t_i=0$}
\end{subfigure}
\hspace{-15pt}
\begin{subfigure}{0.265\textwidth}
\includegraphics[width=1\linewidth]{comparetruevalue1over3.pdf}
\caption{$t_i=1/3$}
\end{subfigure}
\hspace{-15pt}
\begin{subfigure}{0.265\textwidth}
\includegraphics[width=1\linewidth]{comparetruevalue2over3.pdf}
\caption{$t_i=2/3$}
\end{subfigure}
\hspace{-15pt}
\begin{subfigure}{0.265\textwidth}
\includegraphics[width=1\linewidth]{comparetruevalue1.pdf}
\caption{$t_i=1$}
\end{subfigure}

\caption{Different mechanisms' noise variances versus the privacy budget $\epsilon$, when the true input $t_i$ is $0$, $1/3$, $2/3$ or $1$.}
\label{fig:CompareUnderDifferentInputs}
\end{figure*}
}

A natural question is: can we design a perturbation method that \textit{combines} the advantages of the Laplace mechanism and Duchi~et~al.'s solution to minimize the variance of $t^*_i$ across a \textit{wide} spectrum of $\epsilon$? Intuitively, such a method should confine $t^*_i$ to a relatively small domain (as Duchi~et~al.'s solution does), and should allow $t^*_i$ to be close to $t_i$ with reasonably large probability (as the Laplace mechanism does). In what follows, we will present a new perturbation method based on this intuition.


\subsection{Piecewise Mechanism} \label{sec:PM}

\header
Our first proposal, referred to as the {\it Piecewise Mechanism} (PM), takes as input a value $t_i \in [-1, 1]$, and outputs a perturbed value $t^*_i$ in $[-C, C]$, where
$$C = \frac{\exp(\epsilon/2)+1}{\exp(\epsilon/2)-1}.$$
The probability density function (pdf) of $t^*_i$ is a piecewise constant function as follows:  \label{page:Piecewise:Mechanism}
\begin{equation} \label{eqn:FM}
\pdf(t^*_i=x \mid t_i) =
\begin{cases}
p, & \textrm{if $x \in \left[\ell(t_i), \,r(t_i)\right]$,}  \\[5pt]
\frac{p}{\exp(\epsilon)}, & \textrm{if $x \in \left[-C, \,\ell(t_i)\right) \,\mathlarger{\cup} \,\left(r(t_i),\, C\right]$.}
\end{cases}
\end{equation}
where
\begin{align*}
& p  = \frac{\exp(\epsilon) - \exp(\epsilon/2)}{2\exp(\epsilon/2) + 2},   \\[2mm]
& \ell(t_i) = \frac{C+1}{2}\cdot t_i - \frac{C-1}{2}, \textrm{ and} \\[2mm]
& r(t_i)  = \ell(t_i) + C-1.
\end{align*}
Let $\pdf(t^*_i)$ be short for $\pdf(t^*_i=x \mid t_i)$.
Fig.~\ref{fig:exampledistribution} illustrates $\pdf(t^*_i)$ for the cases of $t_i = 0$, $t_i = 0.5$, and $t_i = 1$. Observe that when $t_i = 0$, $\pdf(t^*_i)$ is symmetric and consists of three ``pieces'', among which the center piece (i.e., $t^*_i \in [\ell(t_i), r(t_i)]$) has a higher probability than the other two. When $t_i$ increases from $0$ to $1$, the length of the center piece remains unchanged (since $r(t_i) - \ell(t_i) = C-1$), but the length of the rightmost piece (i.e., $t^*_i \in (r(t_i), C]$) decreases, and is reduced to $0$ when $t_i = 1$. The case when $t_i < 0$ can be illustrated in a similar manner.

\begin{figure*}
  \centering
  \begin{tabular}{ccc}
  \multicolumn{3}{c}{}\\
    \hspace{0mm}\includegraphics[width=0.25\textwidth]{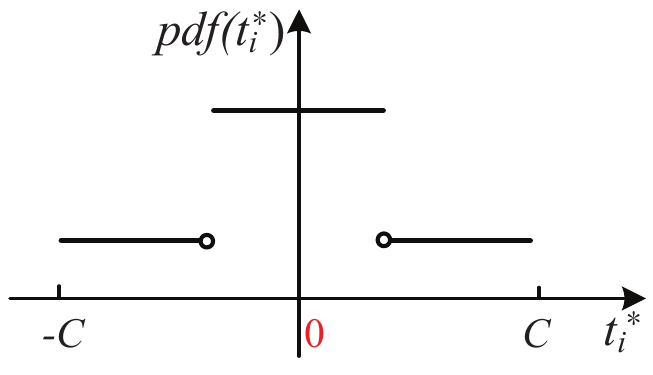} &
    \hspace{8mm}\includegraphics[width=0.25\textwidth]{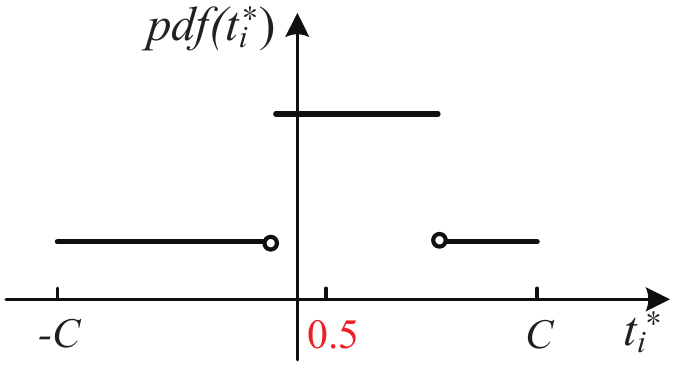} &
    \hspace{8mm}\includegraphics[width=0.25\textwidth]{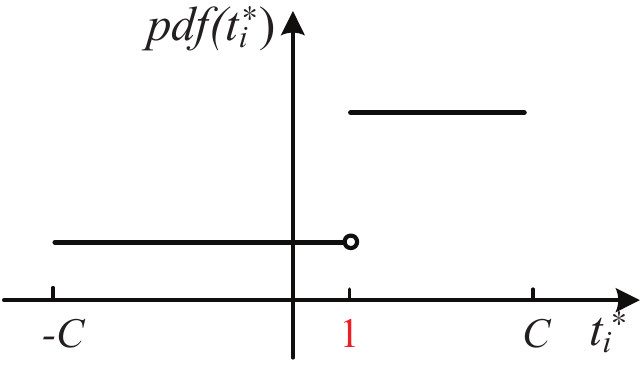}
    \\
    (a) $t_i = 0$ & (b) $t_i = 0.5$ & (c) $t_i = 1$
   \end{tabular}
  \vspace{0pt}
  \caption{The noisy output $t^*_i$'s probability density function $\pdf(t^*_i)$ in the Piecewise Mechanism.\vspace{-10pt}}
  \label{fig:exampledistribution} 
\end{figure*}


\begin{algorithm}[t]
\caption{Piecewise Mechanism for One-Dimensional Numeric Data.}\label{alg:FM}
\SetKwInOut{Input}{input}
\SetKwInOut{Output}{output}
\Input{tuple $t_i \in [-1,\,1]$ and privacy parameter $\epsilon.$}
\Output{tuple $t^*_i \in \left[-C,\, C\right].$}
Sample $x$ uniformly at random from $[0, 1]$\;
\If{$x < \frac{e^{\epsilon/2}}{e^{\epsilon/2}+1}$}
{
    Sample $t^*_i$ uniformly at random from $\left[\ell(t_i), \,r(t_i)\right]$\;
}
\Else
{
    Sample $t^*_i$ uniformly at random from $\left[-C, \,\ell(t_i)\right) \,\mathlarger{\cup} \,\left(r(t_i),\, C\right]$\;
}
\Return $t^*_i$
\end{algorithm}

Algorithm~\ref{alg:FM} shows the pseudo-code of PM, assuming the input domain is $[-1, 1]$. 
In general, when the input domain is $t_i \in [-r, r], r>0$, the user (i) computes $t_i' = t_i/r$, (ii) perturbs $t_i'$ using PM, since $t_i' \in [-1, 1]$, and (iii) submits $r \cdot t_i^*$ to the server, where $t_i^*$ denotes the noisy value output by Algorithm \ref{alg:FM}. It can be verified that $r \cdot t_i^*$ is an unbiased estimator of $t_i$. The above method requires that the user knows $r$, which is a common assumption in the literature, e.g., in Duchi et al.'s work \cite{DuchiJW18}.

The following lemmas establish the theoretical guarantees of Algorithm~\ref{alg:FM}.

\begin{lemma} \label{lmm:basic-privacy}
Algorithm~\ref{alg:FM} satisfies $\epsilon$-local differential privacy. In addition, given an input value $t_i$, it returns a noisy value $t^*_i$ with $\E[t^*_i] = t_i$ and
$$\Var[t^*_i] = \frac{{t_i}^2}{e^{\epsilon/2}-1}+\frac{e^{\epsilon/2}+3}{3(e^{\epsilon/2} - 1)^2}.$$
\end{lemma}
The proof appears in the full version \cite{tr}.

By Lemma~\ref{lmm:basic-privacy}, PM returns a noisy value $t^*_i$ whose variance is at most
$$\frac{1}{e^{\epsilon/2}-1}+\frac{e^{\epsilon/2}+3}{3(e^{\epsilon/2} - 1)^2} = \frac{4e^{\epsilon/2}}{3(e^{\epsilon/2}-1)^2}.$$
The purple line in Fig.~\ref{fig:FMCom} illustrates this worst-case variance of PM as a function of $\epsilon$. Observe that PM's worst-case variance is considerably smaller than that of  Duchi~et~al.'s solution when $\epsilon \ge 1.29$, and is only slightly larger than the latter when $\epsilon  < 1.29$, where 1.29 is $x$-coordinate of the point that the Duchi et al.' solution curve intersects that of PM in Fig. \ref{fig:FMCom}. Furthermore, it can be proved that PM's worst-case variance is strictly smaller than Laplace mechanism's, regardless of the value of $\epsilon$. This makes PM be a more preferable choice than both the Laplace mechanism and Duchi~et~al.'s solution.

Furthermore, Lemma~\ref{lmm:basic-privacy} also shows that the variance of $t^*_i$ in PM monotonically decreases with the decrease of $|t_i|$, which makes PM particularly effective when the distribution of the input data is skewed towards small-magnitude values. (In Section~\ref{sec:exp}, we show that $|t_i|$ tends to be small in a large class of applications.) In contrast, Duchi~et~al.'s solution incurs a noise variance that {\it increases} with the decrease of $|t_i|$, as shown in Equation~\ref{eqn:basic-duchi-variance}. 

Now consider the estimator $\frac{1}{n}\sum_i^n t^*_{i}$ used by the aggregator to infer the mean value of all $t_i$. The variance of this estimator is $1/n$ of the average variance of $t^*_i$. Based on this, the following lemma establishes the accuracy guarantee of $\frac{1}{n}\sum_i^n t^*_{i}$.
\begin{lemma} \label{lmm:basic-accuracy}
 Let $Z = \frac{1}{n} \sum_{i=1}^n t^*_i$ and $X = \frac{1}{n} \sum_{i=1}^n t_i$. With \mbox{at least}  $1 - \beta$ probability,
$$ \big|Z - X\big| = O\left(\frac{\sqrt{\log(1/\beta)}}{\epsilon \sqrt{n}}\right).$$
\end{lemma}
We omit the proof of Lemma~\ref{lmm:basic-accuracy} as it is a special case of Lemma~\ref{lmm:multi-accuracy} to be presented in Section~\ref{sec:multi-numeric}.

\header
\textbf{Remark.}
PM bears some similarities to SCDF~\cite{Soria-ComasD13} and Staircase mechanism~\cite{GengKOV15} described in Section \ref{sec:numeric_existing}, in the sense that the added noise in PM also follows a piece-wise constant distribution, as in SCDF and Staircase. On the other hand, there are two crucial differences between PM and SCDF/Staircase. First, SCDF and Staircase mechanism assume an unbounded input, and produce an unbounded output (i.e., with range $(-\infty, +\infty)$) accordingly. In contrast, PM has both bounded input (with domain $[-1, 1]$) and output (with range $[-C, C]$). Second, the noise distribution of SCDF/Staircase consists of an infinite number of ``pieces'' that are data {\em independent}, whereas the output distribution of the piecewise mechanism consists of up to three ``pieces'' whose lengths and positions depend on the input data.

\subsection{Hybrid Mechanism} \label{sec:HM}

As discussed in Section \ref{sec:PM}, the worst-case result accuracy of PM dominates that of the Laplace mechanism, and yet it can still be (slightly) worse than Duchi et al's solution, since the noise variance incurred by the former (resp.\ latter) decreases (resp.\ increases) with the decrease of $|t_i|$. Can we construct a method that that preserves the advantages of PM and is at the same time always no worse than Duchi et al's solution? The answer turns out to be positive: that we can combine PM and Duchi et al's solution into a new \emph{Hybrid Mechanism} (\emph{HM}). Further, the combination used in HM is non-trivial; as a result, the noise variance of HM is often smaller than both PM and Duchi et al's solution, as shown in Fig.~\ref{fig:FMCom} on Page~\pageref{fig:FMCom}.

In particular, given an input value $t_i$, HM flips a coin whose head probability equals a constant $\alpha$; if the coin shows a head (resp.\ tail), then we invoke PM (resp.\ Duchi~et~al.'s solution) to perturb $t_i$. \label{page:Hybrid:Mechanism} Given $t_i$ and $\epsilon$, the noise variance incurred by HM is
\begin{align}
\sigma^2_{H}(t_i, \epsilon) = \alpha \cdot \sigma^2_{P}(t_i, \epsilon) + (1 - \alpha) \cdot \sigma^2_{D}(t_i, \epsilon), \nonumber
\end{align}
where $\sigma^2_{P}(t_i, \epsilon)$ and $\sigma^2_{D}(t_i, \epsilon)$ denote the noise variance incurred by PM and Duchi et al.'s solution, respectively, when given $t_i$ and $\epsilon$ as input. We have the following lemma.
\begin{lemma} \label{lmm:numeric-hybrid-min}
Let $\epsilon^*$ be defined as:
\begin{align} \label{eqn:numeric-eps-star}
\epsilon^* & = \textstyle  \ln \left(  \frac{-5 + 2\sqrt[3]{6353 - 405 \sqrt{241}} \, + \, 2\sqrt[3]{6353 + 405 \sqrt{241}}}{27} \right)
\approx 0.61.
\end{align}
The term $\max_{t_i \in [-1, 1]}\sigma^2_H(t_i, \epsilon)$ is minimized when
\begin{align}
\alpha = \begin{cases}
1-e^{-\epsilon/2}, & \text{for }\epsilon > \epsilon^*, \\ 0, & \text{for } \epsilon \leq \epsilon^*.
\end{cases} \label{eq-HM-alpha}
\end{align}
\end{lemma}

The proof appears in the full version \cite{tr}.

By Lemma~\ref{lmm:numeric-hybrid-min}, when $\alpha$ satisfies Equation~\ref{eq-HM-alpha}, the worst-case noise variance of HM is:
\begin{align} \label{eqn:HM-duchi-variance}
& \max_{t_i \in [-1, 1]}\sigma^2_H(t_i, \epsilon) = \begin{cases}
\frac{e^{\epsilon/2}+3}{3e^{\epsilon/2}(e^{\epsilon/2}-1)} + \frac{(e^{\epsilon}+1)^2}{e^{\epsilon/2}(e^{\epsilon}-1)^2} , & \hspace{-8pt}\text{for }\epsilon > \epsilon^*, \\[2mm] \left(\frac{e^\epsilon+1}{e^\epsilon-1}\right)^2, & \hspace{-8pt}\text{for } \epsilon \leq \epsilon^*.
\end{cases}
\end{align}
Based on Equation~\ref{eqn:HM-duchi-variance},  Lemma~\ref{lmm:basic-privacy}, and Equation~\ref{eqn:basic-duchi-variance}, which present $\sigma^2_H(t_i, \epsilon)$, $\sigma^2_P(t_i, \epsilon)$, and $\sigma^2_D(t_i, \epsilon)$, respectively, we can show that HM often dominates both PM and Duchi~et~al.'s solution in minimizing the worst-case noise variance. The detailed results are summarized under $d=1$ (meaning one dimension) in Table~\ref{table-comparison} of Section~\ref{sec:intro}, where $\epsilon^*$ follows from Equation~\ref{eqn:numeric-eps-star} and $\epsilon^{\#}$ is derived by solving $\epsilon$ which makes $\max_{t_i \in [-1, 1]}\sigma^2_P(t_i, \epsilon)$ and $\max_{t_i \in [-1, 1]}\sigma^2_D(t_i, \epsilon)$ equal. We highlight  some results as follows.
\begin{corollary} \label{Cor1}
Suppose that $\alpha$ satisfies Equation~\ref{eq-HM-alpha}. If $\epsilon > \epsilon^*$,
$$\max_{t_i \in [-1, 1]}\hspace{-1pt}\sigma^2_H(t_i, \epsilon) \hspace{-1pt}<\hspace{-1pt} \min\hspace{-2pt}\left\{\hspace{-1pt}\max_{t_i \in [-1, 1]}\hspace{-1pt}\sigma^2_P(t_i, \epsilon), \max_{t_i \in [-1, 1]}\hspace{-1pt}\sigma^2_D(t_i, \epsilon)\hspace{-1pt}\right\}\hspace{-2pt};$$
otherwise,
$$\max_{t_i \in [-1, 1]}\sigma^2_H(t_i, \epsilon) = \max_{t_i \in [-1, 1]}\sigma^2_D(t_i, \epsilon) < \max_{t_i \in [-1, 1]}\sigma^2_P(t_i, \epsilon).$$
\end{corollary}

The red line in Fig.~\ref{fig:FMCom} on Page~\pageref{fig:FMCom} shows the worst-case noise variance incurred by HM, which is consistently no higher than those of all other three methods (HM reduces to Duchi et al.'s solution for $\epsilon \leq \epsilon^*$). In addition, observe that PM's accuracy is close to HM's, which demonstrates the effectiveness of PM.



\section{Collecting Multiple Attributes} \label{sec:multi}

We now consider the case where each user's data record contains $d>1$ attributes. In this case, a straightforward solution is to collect each attribute separately using a single-attribute perturbation algorithm, such that every attribute is given a privacy budget $\epsilon/d$. Then, by the composition theorem~\cite{DworkR14}, the collection of all attributes satisfies $\epsilon$-LDP. This solution, however, offers inferior data utility. For example, suppose that all $d$ attributes are numeric, and we process each attribute using PM, setting the privacy budget to $\epsilon/d$. Then, by Lemma~\ref{lmm:basic-accuracy}, the amount of noise in the estimated mean of each attribute is $O\left(\frac{d\sqrt{\log d}}{\epsilon \sqrt{n}}\right)$, which is super-linear to $d$, and hence, can be excessive when $d$ is large. To address the problem, the first and only existing solution that we are aware of is by Duchi~et~al.~\cite{DuchiJW18} for the case of multiple numeric attributes, presented in Section~\ref{sec:multi-duchi}.

\subsection{Existing Solution for Multiple Numeric Attributes} \label{sec:multi-duchi}

Algorithm~\ref{alg:duchimulti} shows the pseudo-code of Duchi et al.'s solution for multidimensional numeric data.
It takes as input a tuple $t_i \in [-1, 1]^d$ of user $u_i$ and a privacy parameter $\epsilon$, and outputs a perturbed vector $t^*_i \in \{-B, B\}^d$, where $B$ is a constant decided by $d$ and $\epsilon$. Upon receiving the perturbed tuples, the aggregator simply computes the average value for each attribute over all users, and outputs these averages as the estimates of the mean values for their corresponding attributes. Next, we focus on the calculation of $B$, which is rather complicated.

\begin{algorithm}[t]
\caption{Duchi et al.'s Solution~\cite{DuchiJW18} for \mbox{Multidimensional} Numeric Data.}\label{alg:duchimulti}
\SetKwInOut{Input}{input}
\SetKwInOut{Output}{output}
\Input{tuple $t_i \in [-1,1]^d$ and privacy parameter $\epsilon.$}
\Output{tuple $t^*_i \in \{-B, B\}^d.$}
    Generate a random tuple $v \in \{-1, 1\}^d$ by sampling each $v[A_j]$ independently from the following distribution:
    \parbox{68mm}{\begin{equation*}
    \Pr[v[A_j] = x] = \begin{cases}
        \frac{1}{2} + \frac{1}{2} t_i[A_j], & \text{if $x = 1$}  \\[2mm]
        \frac{1}{2} - \frac{1}{2} t_i[A_j], & \text{if $x = -1$}\\
    \end{cases}
    \end{equation*}}\;
    Let $T^+$ (resp.\ $T^-$) be the set of all tuples $t^* \in \{-B, B\}^d$ such that $t^* \cdot v \geq 0$ (resp.\ $t^* \cdot v \le 0$)\;
    Sample a Bernoulli variable $u$ that equals $1$ with ${e^\epsilon / (e^\epsilon + 1)}$ probability\;
    \If{$u = 1$}
    {
        \Return a tuple uniformly at random from $T^+$\;
    }
    \Else
    {
        \Return a tuple uniformly at random from $T^-$\;
    }
\end{algorithm}

Essentially, $B$ is a scaling factor to ensure that the expected value of a perturbed attribute is the same as that of the exact attribute value. First, we compute:

\begin{equation} \label{eqn:basic-Cd}
C_d = 
\begin{cases}
\frac{2^{d-1}}{{d-1 \choose {(d-1)/2} }}, &\text{if  $d$ is odd,}\\[4mm]
\frac{2^{d-1} + \frac{1}{2} {d \choose {d/2} } }{{d-1 \choose {d/2} }}
, &\text{otherwise.}
\end{cases}
\end{equation}
Then, $B$ is calculated by:
\begin{equation} \label{eqn:basic-B}
B = \frac{\exp(\epsilon) + 1}{\exp(\epsilon) - 1} \cdot C_d.
\end{equation}
Duchi et al.\ show that $\frac{1}{n} \sum_{i = 1}^n t^*_i[A_j]$ is an unbiased estimator of the mean of $A_j$, and
\begin{equation} \label{eqn:basic-error}
\E\left[\max_{j \in [1, d]} \left|\frac{1}{n} \sum_{i = 1}^n t^*_i[A_j] - \frac{1}{n} \sum_{i = 1}^n t_i[A_j]\right|\,\right] = O\left(\frac{\sqrt{d \log d}}{\epsilon \sqrt{n}}\right),
\end{equation}
which is asymptotically optimal \cite{DuchiJW18}.

Although Duchi et al's method can provide strong privacy assurance and asymptotic error bound, it is rather sophisticated, and it cannot handle the case that a tuple contains numeric attributes as well as categorical attributes. To address this issue, we present extensions of PM and HM that (i) are much simpler than Duchi et al.'s solution but achieve the same privacy assurance and asymptotic error bound, and (ii) can handle any combination of numeric and categorical attributes. For ease of exposition, we first extend PM and HM for the case when each $t_i$ contains only numeric attributes in Section~\ref{sec:multi-numeric}, and then discuss the case of arbitrary attributes in Section~\ref{sec:multi-categorical}.

\subsection{Extending PM and HM for Multiple Numeric Attributes} \label{sec:multi-numeric}


\begin{algorithm}[t]
\caption{Our Method for Multiple Numeric \mbox{Attributes}.}\label{alg:multinumeric}
\SetKwInOut{Input}{input}
\SetKwInOut{Output}{output}
\Input{tuple $t_i \in [-1,1]^d$ and privacy parameter $\epsilon$.}
\Output{tuple $t^*_i \in [-C\cdot d, \:\: C\cdot d]^d.$}
    Let $t^*_i = \langle 0, 0, \ldots, 0\rangle$\;
    Let $k = \max\left\{1, \min\left\{d, \, \left\lfloor \frac{\epsilon}{2.5}\right\rfloor\right\}\right\}$\;
    Sample $k$ values uniformly without replacement from $\{1, 2, \ldots, d\}$\; \label{alg:multinumeric2:line:sample}
    \For{each sampled value j}
    {Feed $t_i[A_j]$ and $\frac{\epsilon}{k}$ as input to PM or HM, and obtain a noisy value $x_{i,j}$\; \label{alg:multinumeric2:line:PM:HM}
    $t^{*}_i[A_j] = \frac{d}{k}   x_{i,j}$\;
    }
    \Return $t^*_i$
\end{algorithm}

Algorithm~\ref{alg:multinumeric} shows the pseudo-code of our extension of PM and HM for multidimensional numeric data. Given a tuple $t_i \in [-1, 1]^d$, the algorithm returns a perturbed tuple $t^*_i$ that has non-zero value on $k$ attributes, where
\begin{equation} \label{eqn:multi-k}
k = \max\left\{1, \min\left\{d, \, \left\lfloor \frac{\epsilon}{2.5}\right\rfloor\right\}\right\}.
\end{equation}
In particular, each $A_j$ of those $k$ attributes is selected uniformly at random (without replacement) from all $d$ attributes of $t_i$, and $t^*_i[A_j]$ is set to $\frac{d}{k} \cdot x$, where $x$ is generated by PM or HM given $t_i[A_j]$ and $\frac{\epsilon}{k}$ as input. 

The intuition of Algorithm~\ref{alg:multinumeric} is as follows. By requiring each user to submit $k$ (instead of $d$) attributes, it increases the privacy budget for each attribute from $\epsilon/d$ to $\epsilon/k$, which in turn reduces the noise variance incurred. As a trade-off, sampling $k$ out of $d$ attributes entails additional estimation error, but this trade-off can be balanced by setting $k$ to an appropriate value, which is shown in Equation~\ref{eqn:multi-k}. We derive the setting of $k$ by minimizing the worst-case noise variance of Algorithm~\ref{alg:multinumeric} when it utilizes PM (resp.\ HM)\footnote{We discuss how to obtain the value of $k$ in the full version \cite{tr}.}.



\begin{lemma} \label{lmm:multi-privacy}
Algorithm~\ref{alg:multinumeric} satisfies $\epsilon$-local differential privacy. In addition, given an input tuple $t_i$, it outputs a noisy tuple $t^*_i$, such that for any $j \in [1, d]$, $\E\left[t^{*}_i[A_j]\right] = t_i[A_j]$.
\end{lemma}

The proof appears in the full version \cite{tr}.

By Lemma~\ref{lmm:multi-privacy}, the aggregator can use $\frac{1}{n} \sum_{i = 1}^n t^*[A_j]$ as an unbiased estimator of the mean of $A_j$. The following lemma shows that the accuracy guarantee of this estimator matches that of Duchi et al.'s solution for multidimensional numeric data (see Equation~\ref{eqn:basic-error}), which has been proved to be asymptotically optimal \cite{DuchiJW18}. This indicates that Algorithm~\ref{alg:multinumeric}'s accuracy guarantee is also optimal in the asymptotic sense.


\begin{figure*}
  \centering
  \footnotesize
  \begin{tabular}{cccc}
  \multicolumn{4}{c}{\includegraphics[width=0.65\textwidth]{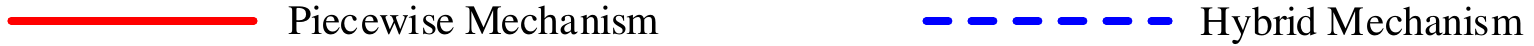}} \vspace{-2pt} \\
    \hspace{-3mm}\includegraphics[width=0.23\textwidth]{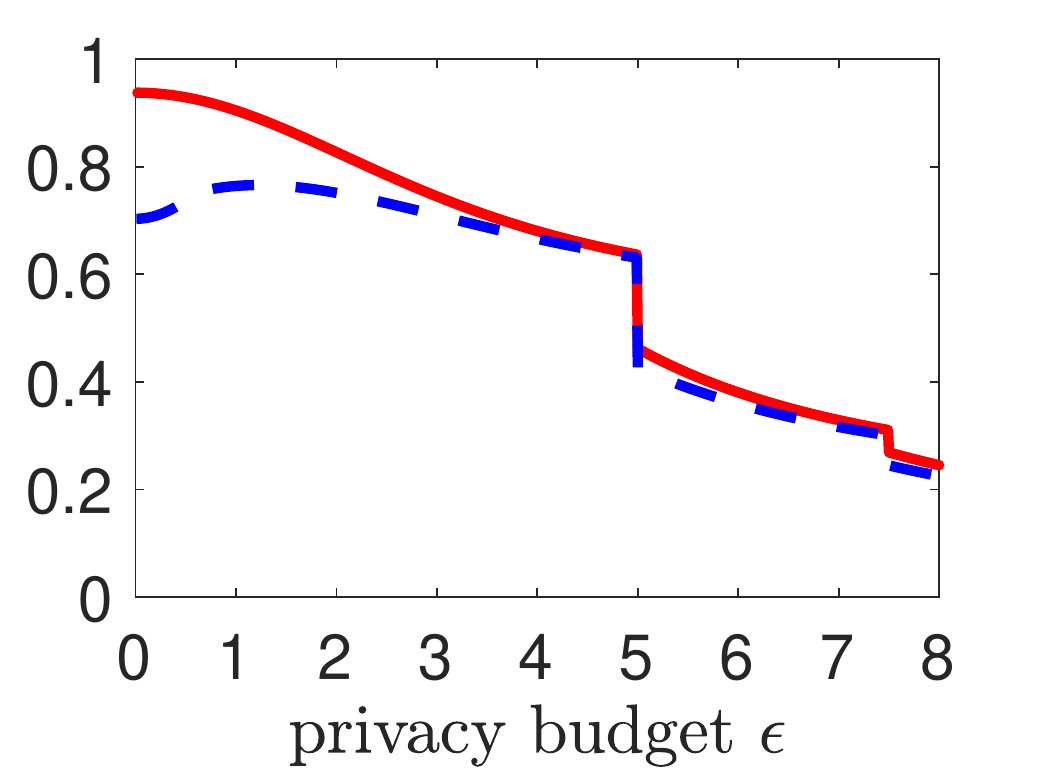} &
    \hspace{-4mm}\includegraphics[width=0.23\textwidth]{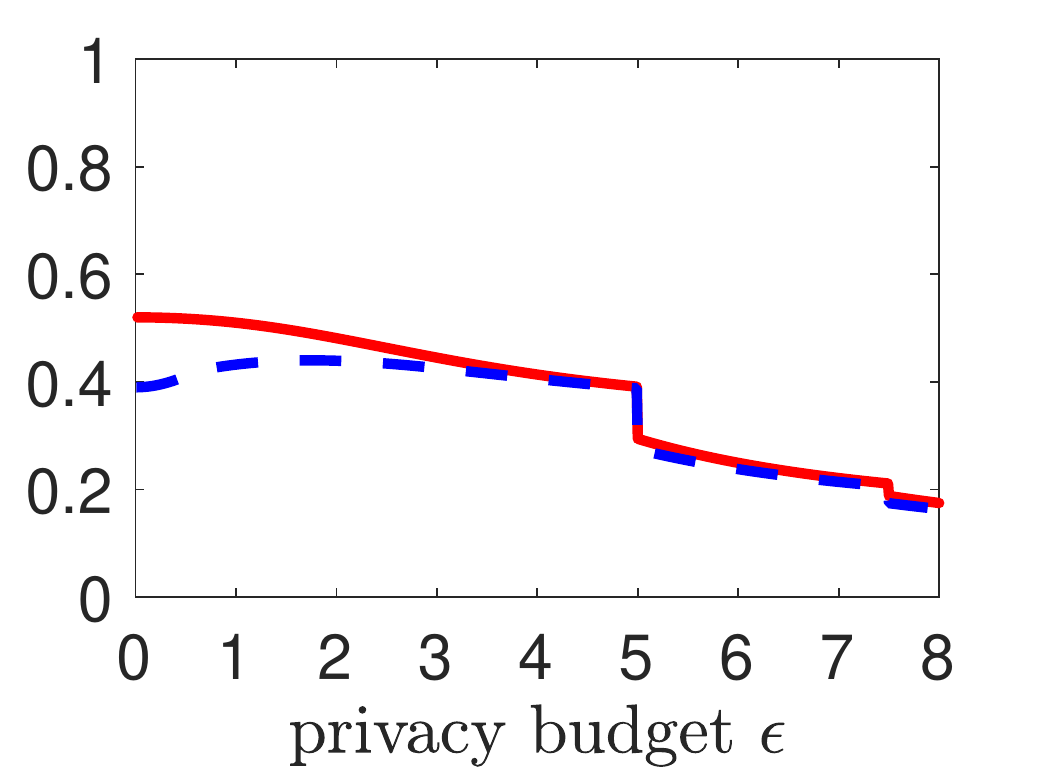} &
    \hspace{-4mm}\includegraphics[width=0.23\textwidth]{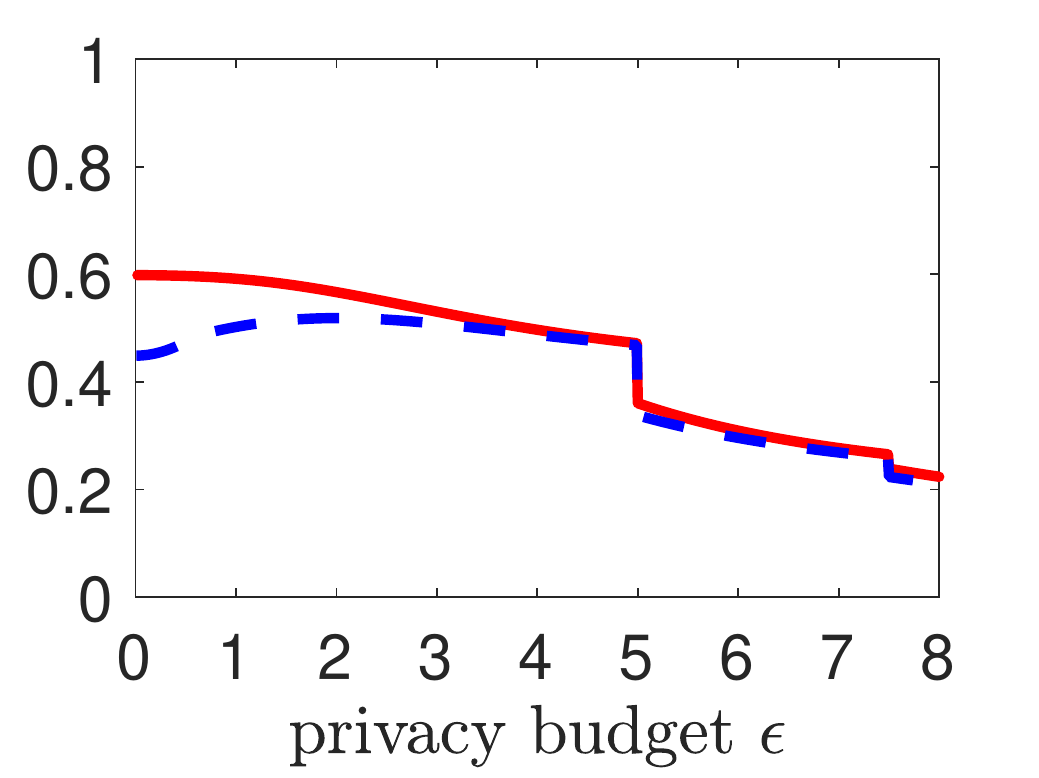} &
    \hspace{-4mm}\includegraphics[width=0.23\textwidth]{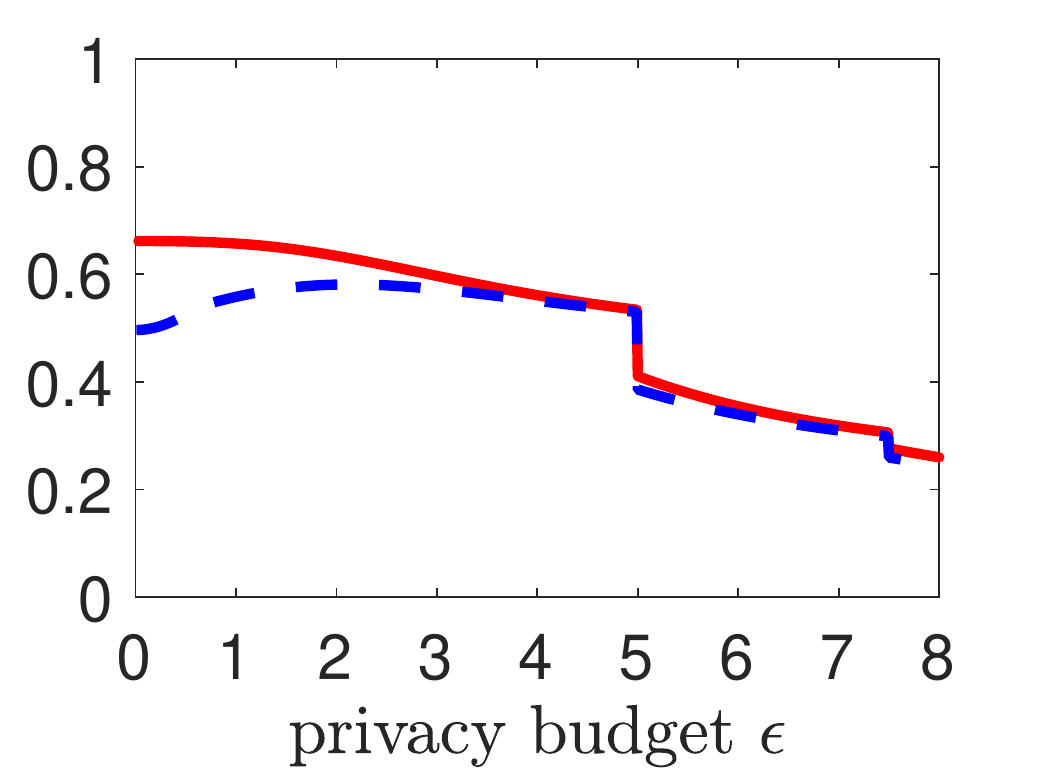} \\
    (a) $d=5$ & (b) $d=10$ & (c) $d=20$ & (d) $d=40$
   \end{tabular}
 \vspace{-6pt} \caption{The worst-case variance of PM (resp.\ HM) as a fraction of the worst-case variance of Duchi et al.'s solution, for various privacy budget $\epsilon$ and data dimensionality $d$.\vspace{-4mm}  
}
  \label{fig:CompareUnderMultipleDimensions}
  \end{figure*}

\begin{lemma} \label{lmm:multi-accuracy}
For any $j \in [1, d]$, let $Z[A_j] = \frac{1}{n} \sum_{i=1}^n t^*_i[A_j]$ and $X[A_j] = \frac{1}{n} \sum_{i=1}^n t_i[A_j]$. With at least $1 - \beta$ probability,
$$\max_{j \in [1, d]} \big|Z[A_j] - X[A_j]\big| = O\left(\frac{\sqrt{d \log(d/\beta)}}{\epsilon \sqrt{n}}\right).$$
\end{lemma}

The proof appears in the full version \cite{tr}.

Lemma \ref{lmm:multi-dimension-variance} discusses the noise variances induced by Duchi~et~al.'s solution, PM and HM, respectively.

\begin{lemma} \label{lmm:multi-dimension-variance}

For a $d$-dimensional numeric tuple $t_i$ which is perturbed as $t^{*}_i$ under $\epsilon$-LDP, and for each $A_j$ of the $d$ attributes, the variance of $t^{*}_i[A_j]$ induced by Duchi~et~al.'s solution is
\begin{align} 
 & \textstyle   
\Var_D\big[t^{*}_i[A_j]\big] = \left( \frac{e^{\epsilon} + 1}{e^{\epsilon} - 1} \right)^2 {C_d}^2  -  \left(t_i[A_j]\right)^2,  \label{Duchi-variance-multi-dimension}
\end{align}
where $C_d $ 
 is defined by Equation~\ref{eqn:basic-Cd}. Meanwhile,
the variance of $t^{*}_i[A_j]$ induced by PM is
\begin{align} 
 & \textstyle  \Var_P\big[t^{*}_i[A_j]\big] =
 \frac{d(e^{\epsilon/(2k)}+3)}{3k(e^{\epsilon/(2k)} - 1)^2}   +  \left[ \frac{d\cdot e^{\epsilon/(2k)}}{k(e^{\epsilon/(2k)} - 1) } -1\right] \left(t_i[A_j]\right)^2;
  \label{PM-variance-multi-dimension}
\end{align}
and the variance of $t^{*}_i[A_j]$ induced by HM is
\begin{align} 
 & \textstyle  \Var_H\big[t^{*}_i[A_j]\big]
= \nonumber \\ &
 \begin{cases}
\hspace{-1pt}\frac{d}{k} \hspace{-1pt} \left[ \frac{e^{\epsilon/(2k)}+3}{3e^{\epsilon/(2k)}(e^{\epsilon/(2k)}-1)} \hspace{-1pt}+ \hspace{-1pt}\frac{(e^{\epsilon/k}+1)^2}{e^{\epsilon/(2k)}(e^{\epsilon/k}-1)^2} \right] \hspace{-1pt}+\hspace{-1pt}  \left(\frac{d}{k}  \hspace{-1pt}-\hspace{-1pt} 1 \right ) \hspace{-1pt}\left(t_i[A_j]\right)^2 , & ~ \\ ~ & \hspace{-100pt} \text{for }\epsilon/k > \epsilon^*,  \\ \hspace{-1pt}\frac{d}{k} \left(\frac{e^{\epsilon/k}+1}{e^{\epsilon/k}-1}\right)^2 +  \left(\frac{d}{k}  - 1 \right ) \left(t_i[A_j]\right)^2 , &  \hspace{-100pt}  \text{for } \epsilon/k \leq \epsilon^*,
\end{cases} \nonumber \\[-20pt] & ~
 \label{HM-variance-multi-dimension}
\end{align}
where $\epsilon^*$ is defined by Equation~\ref{eqn:numeric-eps-star}.
\end{lemma}
\eat{
\begin{proof}
Applying Equations~\ref{eq:multiple:variance:eq1} and~\ref{eq:multiple:variance:eq1:HM} to Equation~\ref{eq:multiple:variance:eq1:utility}, we can get the noise variances of our PM and HM in Equations~\ref{PM-variance-multi-dimension} and~\ref{HM-variance-multi-dimension} in this lemma. Also, in Duchi et al.'s solution given by Algorithm~\ref{alg:duchimulti}, since $(t^{*}_i[A_j])^2$ always equals $B^2$ (i.e., $\left( \frac{e^{\epsilon} + 1}{e^{\epsilon} - 1} \right)^2 { C_d}^2$), it follows that $\E\big[(t^{*}_i[A_j])^2\big]  = \left( \frac{e^{\epsilon} + 1}{e^{\epsilon} - 1} \right)^2 { C_d}^2$, which along with Duchi~et~al.'s result $\E\big[t^{*}_i[A_j]\big] = t_i[A_j]$ implies Equation~\ref{Duchi-variance-multi-dimension}.
\end{proof}
}

From Equations~\ref{Duchi-variance-multi-dimension},~\ref{PM-variance-multi-dimension}, and~\ref{HM-variance-multi-dimension}, we can prove the following:
\begin{corollary} \label{cor-PMHMdgeq1}
For any $d > 1$ and $\epsilon > 0$, both PM and HM outperform Duchi~et~al.'s solution in minimizing the worst-case noise variance; more specifically, for any $d > 1$ and $\epsilon > 0$,
\begin{align} 
  \textstyle{\max\limits_{t_i[A_j] \in [-1, 1]}   \Var_H\big[t^{*}_i[A_j]\big]} & < \textstyle{\max\limits_{t_i[A_j] \in [-1, 1]}  \Var_P\big[t^{*}_i[A_j]\big]} \nonumber \\ &  < \textstyle{\max\limits_{t_i[A_j] \in [-1, 1]}  \Var_D\big[t^{*}_i[A_j]\big]}.
  \label{relation-variance-multi-dimension}
\end{align}
\end{corollary}

\eat{
To intuitively understand Corollary~\ref{cor-PMHMdgeq1}, we can consider small $\epsilon$ as an example. For small $\epsilon$, $\max\limits_{t_i[A_j] \in [-1, 1]}  \Var_D\big[t^{*}_i[A_j]\big]$ is about $\frac{{C_d}^2}{\epsilon^2}$ from Equation~\ref{Duchi-variance-multi-dimension}. For small $\epsilon$ which induces $k=1$ by Equation~\ref{eqn:multi-k},  $\max\limits_{t_i[A_j] \in [-1, 1]}  \Var_P\big[t^{*}_i[A_j]\big]$ is about $\frac{d}{3}\cdot \frac{4}{3(\epsilon/2)^2} + \frac{d}{\epsilon/2} - 1$ from Equation~\ref{PM-variance-multi-dimension}. Meanwhile, with $\epsilon/k \leq \epsilon^*$, $\max\limits_{t_i[A_j] \in [-1, 1]}  \Var_H\big[t^{*}_i[A_j]\big]$ is about $\frac{d}{\epsilon^2} + d - 1$ from Equation~\ref{HM-variance-multi-dimension}. For small $\epsilon$ and large $d$, we have $\frac{{C_d}^2}{\epsilon^2}< \frac{d}{3}\cdot \frac{4}{3(\epsilon/2)^2} + \frac{d}{\epsilon/2} - 1<\frac{d}{\epsilon^2} + d - 1$, which help understand (\ref{relation-variance-multi-dimension}).
}

To illustrate Corollary~\ref{cor-PMHMdgeq1}, Fig.~\ref{fig:CompareUnderMultipleDimensions} shows the worst-case variance of PM (resp.\ HM) as a fraction of the worst-case variance of Duchi~et~al.'s solution, for various $d$ and privacy budget $\epsilon$. Observe that for $d = 5, 10, 20, 40$, the wort-case variance of HM is at most 77\% of that of Duchi~et~al.'s solution, and PM's worst-case variance is also smaller than the latter. In our experiments, we demonstrate that both HM and PM outperform Duchi~et~al.'s solution in terms of the empirical accuracy for multidimensional numeric data.

\subsection{Handling Categorical Attributes} \label{sec:multi-categorical}

So far our discussion is limited to numeric attributes. Next we extend Algorithm~\ref{alg:multinumeric} to handle data with both numeric and categorical attributes. Recall from Section \ref{sec:prelim} that for each categorical attribute $A$, our objective is to estimate the frequency of each value $v$ in $A$ over all users. 
We note that most existing LDP algorithms (e.g., \cite{RAPPOR2014,BS15,lininghui}) for categorical data are designed for this purpose, albeit limited to a single categorical attribute.

Formally, we assume that we are given an algorithm $f$ that takes an input a privacy budget $\epsilon$ and a one-dimensional tuple $t_i$ with a categorical attribute $A$, and outputs a perturbed tuple $t^*_i$ while ensuring $\epsilon$-LDP. 
\eat{
In addition, for any value $v \in A$, we assume there is an estimator $g$, such that for any random subset $S \subseteq \{1, 2, \ldots n\}$,  $\frac{d}{n}\sum_{i \in S} g(t^*_i, v)$ is an estimator of the frequency of $v$ over all users. 
}
In addition, we assume there is a function $g(x, y)$ that $g(x, y) = 1$ if $x= y$; and 0, otherwise. 
Then for any value $v \in A$, $\frac{1}{n}\sum_{i \in S} g(t^*_i, v)$ is an estimator of the frequency of $v$ over all users, where $S$ is the set of $\{1, 2, \ldots n\}$.
Then, for the general case when $t_i$ contains $d$ (numeric or categorical) attributes $A_1, A_2, \ldots, A_d$, the extended version of Algorithm~\ref{alg:multinumeric} would request each user to perform the following: 
\begin{enumerate}[topsep = 6pt, parsep = 6pt, itemsep = 0pt, leftmargin=22pt]
    \item Sample $k$ values uniformly at random from $\{1, 2, \dots, d\}$, where $k$ is as defined in Equation~\ref{eqn:multi-k};
    \item For each sampled $j$, if $A_j$ is a numerical attribute, then submit a noisy version of $t[A_j]$ computed as in Lines \ref{alg:multinumeric2:line:sample}--\ref{alg:multinumeric2:line:PM:HM} of Algorithm~\ref{alg:multinumeric}; otherwise (i.e., $A_j$ is a categorical attribute), submit $f(t[A_j], \epsilon/k)$, where $f$ can be any existing solution for perturbing a single categorical attribute under $\epsilon$-LDP;

\end{enumerate}
Once the aggregator collects data from all users, she can estimate the mean of each numeric attribute $A$ in the same way as in Algorithm~\ref{alg:multinumeric}.
\eat{
original version:
In addition, for any categorical attribute $A'$ and any value $v$ in the domain of $A$, she can estimate the frequency of $v$ among all users as $\frac{d}{n}\sum_{v^* \in V^*} g(v^*, v)$, where $V^*$ denotes the set of perturbed $A'$ values submitted by users. 
}
In addition, for any categorical attribute $A'$ and any value $v$ in the domain of $A'$, she can estimate the frequency of $v$ among all users as $\frac{d}{kn}\sum_{v^* \in V^*} g(v^*, v)$, where $V^*$ denotes the set of perturbed $A'$ values submitted by users. 
The accuracy of this estimator depends on both $d$ and the accuracy of single-attribute perturbation algorithm used for $A'$. In our experiments, we apply the optimized unary encoding (OUE) protocol of Wang~et~al.~\cite{lininghui} to perturb a single categorical attribute, which is the current state of the art to our knowledge.

\section{Stochastic Gradient Descent under \\Local Differential Privacy} \label{sec:riskmini}


This section investigates building a class of machine learning models under $\epsilon$-LDP that can be expressed as empirical risk minimization, and solved by stochastic gradient descent (SGD). In particular, we focus on three common types of learning tasks: linear regression, logistic regression, and support vector machines (SVM) classification. 



Suppose that each user $u_i$ has a pair $\langle x_i, y_i\rangle$, where $x_i \in [-1, 1]^d$ and $y_i \in [-1, 1]$ (for linear regression) or $y_i \in \{-1, 1\}$ (for logistic regression and SVM classification). Let $\ell(\cdot)$ be a {\em loss function} that maps a $d$-dimensional {\em parameter vector} $\beta$ into a real number, and  is parameterized by $x_i$ and $y_i$. We aim to identify a parameter vector $\beta^*$ such that
$$ \beta^* = \arg\min\limits_{\beta} \left[  \frac{1}{n}\left(\sum\limits_{i=1}^n \ell(\beta; x_i, y_i)\right) + \frac{\lambda}{2} \|\beta\|^2_2\right],$$
where $\lambda > 0$ is a regularization parameter.
We consider three specific loss functions:
\begin{enumerate}[topsep = 2pt, parsep = 2pt, itemsep = 0pt, leftmargin=20pt]
\item Linear regression: $\ell(\beta; x_i, y_i) = (x_i^T \beta - y_i)^2$;

\item Logistic regression: $\ell(\beta; x_i, y_i) = \log\left(1 + e^{-y_ix_i^T\beta}\right)$;

\item SVM (hinge loss): $\ell(\beta; x_i, y_i) = \max\left\{0, 1 -y_i x_i^T\beta\right\}$.
\end{enumerate}
For convenience, we define
$$\ell^\prime(\beta; x_i, y_i) = \ell(\beta; x_i, y_i) + \frac{\lambda}{2} \|\beta\|^2_2.$$

The proposed approach solves $\beta^*$ using SGD, which starts from an initial parameter vector $\beta_0$, and iteratively updates it into $\beta_1, \beta_2, \ldots$ based on the following equation:
$$\beta_{t+1} = \beta_t - \gamma_t \cdot \nabla \ell^\prime(\beta_t; x, y),$$
where $\langle x, y \rangle$ is the data record of a randomly selected user, $\nabla \ell^\prime(\beta_t; x, y)$ is the gradient of $\ell^\prime$ at $\beta_t$, and $\gamma_t$ is called the {\em learning rate} at the $t$-th iteration. The learning rate $\gamma_t$ is commonly set by a function (called the {\em learning schedule}) of the iteration number $t$; a popular learning schedule is $\gamma_t = O(1/\sqrt{t})$. 

In the non-private setting, SGD terminates when the difference between $\beta_{t+1}$ and $\beta_t$ is sufficiently small. Under \mbox{$\epsilon$-LDP}, however, $\nabla \ell^\prime$ is not directly available to the aggregator, and needs to be collected in a private manner. Towards this end, existing studies \cite{HammCCBX15,DuchiJW18} have suggested that the aggregator asks the selected user in each iteration to submit a noisy version of $\nabla \ell^\prime$, by using the Laplace mechanism or Duchi~et~al.'s solution (i.e., Algorithm~\ref{alg:duchimulti}). Our baseline approach is based on this idea, and improves these existing methods by perturbing $\nabla \ell^\prime$ using Algorithm~\ref{alg:multinumeric}. In particular, in each iteration, we involve a group $G$ of users, and ask each of them to submit a noisy version of the gradient using Algorithm \ref{alg:multinumeric}. 
Here, if any entry of $\nabla \ell_i$ is greater than 1 (resp. smaller than $-1$), then the user should clip it to 1 (resp. $-1$) before perturbation, where $\nabla \ell_i$ is the gradient generated by the $i$-th user in group $G$. That is a common technique referred to as ``gradient clipping'' in the deep learning literature. 
After that, we update the parameter vector $\beta_t$ with the mean of the noisy gradients, i.e.,
$$\textstyle \beta_{t+1} = \beta_t - \gamma_t \cdot \frac{1}{|G|}\sum_{i=1}^{|G|} \nabla \ell^*_i,$$
where $\nabla \ell^*_i$ is the noisy gradient submitted by the $i$-th user in group $G$. This helps because the amount of noise in the average gradient is $O\left(\frac{\sqrt{d\log d}}{\epsilon\sqrt{|G|}}\right)$, which could be acceptable if $|G| = \Omega\left(d(\log d)/\epsilon^2\right)$.


Note that in the non-private case, the aggregator often allows each user to participate in multiple iterations (say $m$ iterations) to improve the accuracy of the model. But it does not work in the local differential privacy setting. To explain this, suppose that the $i$-th ($i \in [1, m]$) gradient returned by the user satisfies \mbox{$\epsilon_i$-differential} privacy. By the composition property of differential privacy \cite{McSherryT07}, if we  enforce $\epsilon$-differential privacy for the user's data, we should have $\sum_{i=1}^m \epsilon_i \le \epsilon$. Consider that we set $\epsilon_i = \epsilon/m$. Then, the amount of noise in each gradient becomes $O\left(\frac{m \sqrt{d \log d}}{\epsilon}\right)$; \vspace{1pt} accordingly, the group size becomes $|G| = \Omega\left(m^2 d \log d/\epsilon^2\right)$, which is $m^2$ times larger compared to the case where each user only participates in at most one iteration. It then follows that the total number of iterations in the algorithm is inverse proportional to $1/m$; i.e., setting $m > 1$ only degrades the performance of the algorithm.

\section{Experiments} \label{sec:exp}

\begin{figure*}
  \centering
  \footnotesize
  \begin{tabular}{cccc}
  \multicolumn{4}{c}{}\\
    \hspace{-3mm}\includegraphics[width=0.23\textwidth]{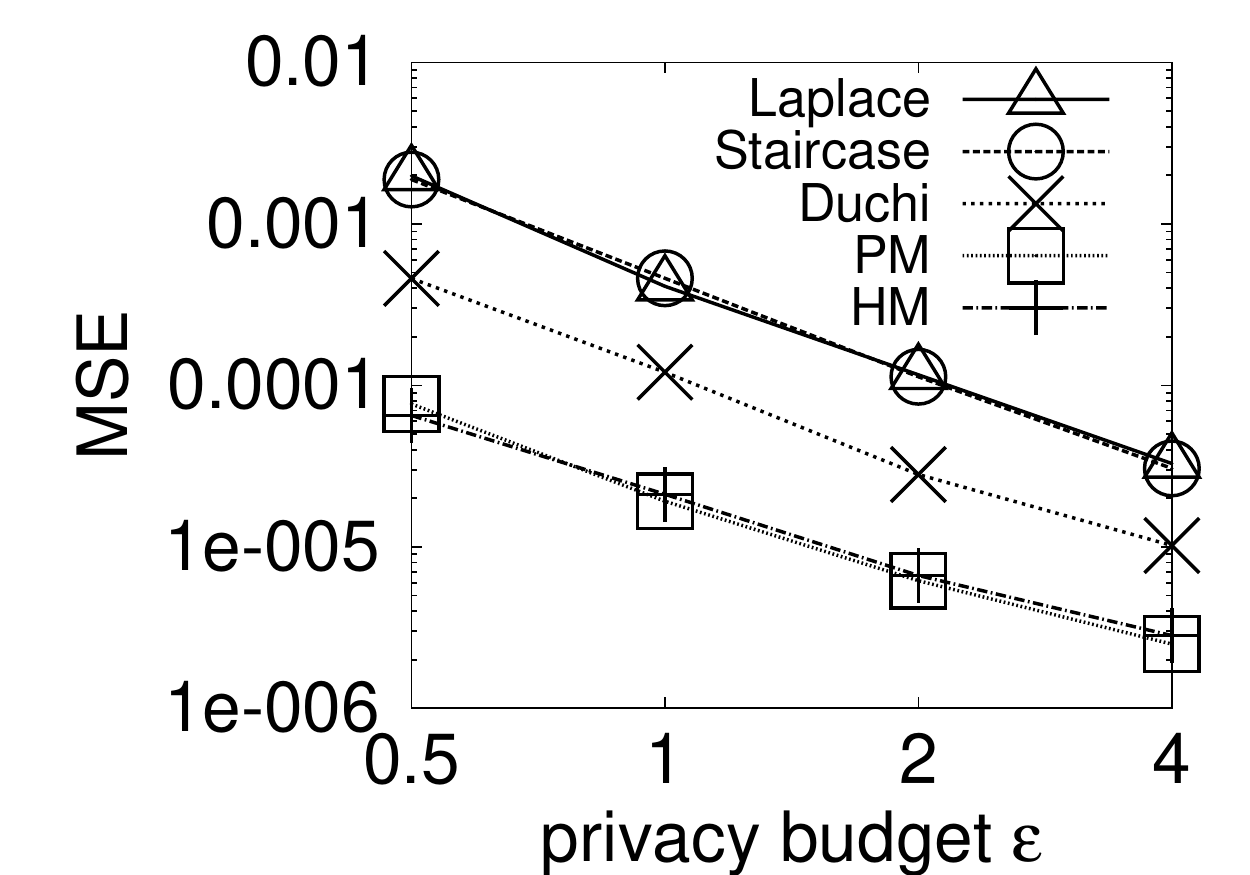} &
    \hspace{-4mm}\includegraphics[width=0.23\textwidth]{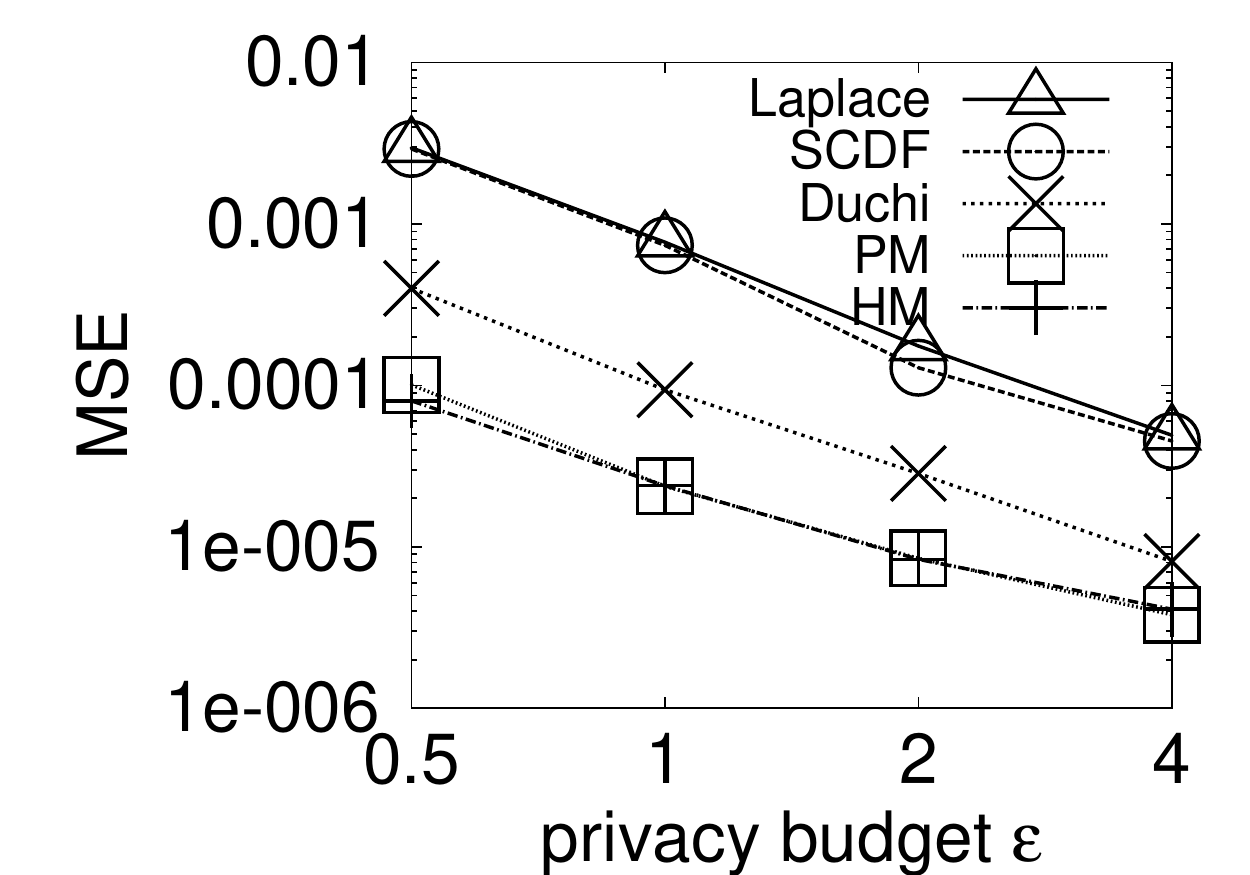} &
    \hspace{-4mm}\includegraphics[width=0.23\textwidth]{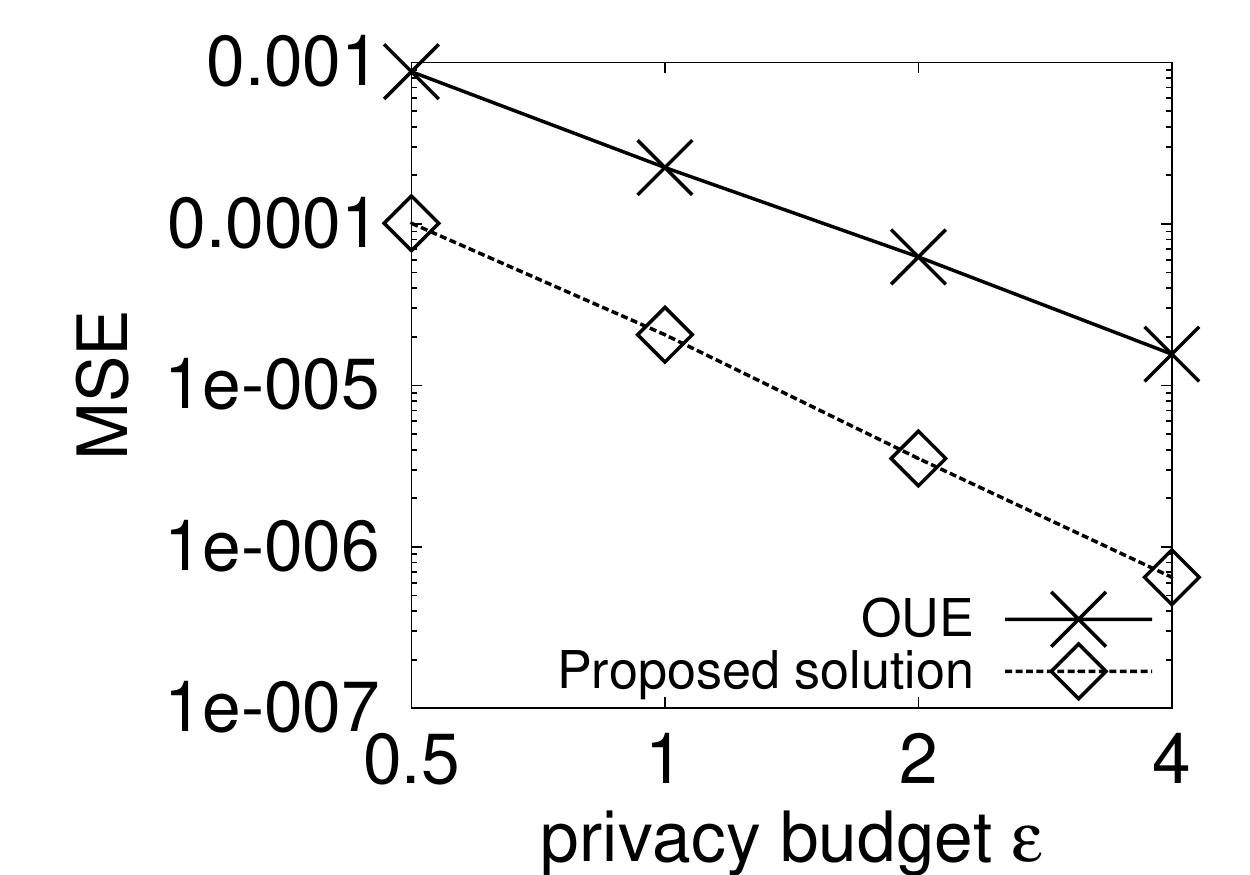} &
    \hspace{-4mm}\includegraphics[width=0.23\textwidth]{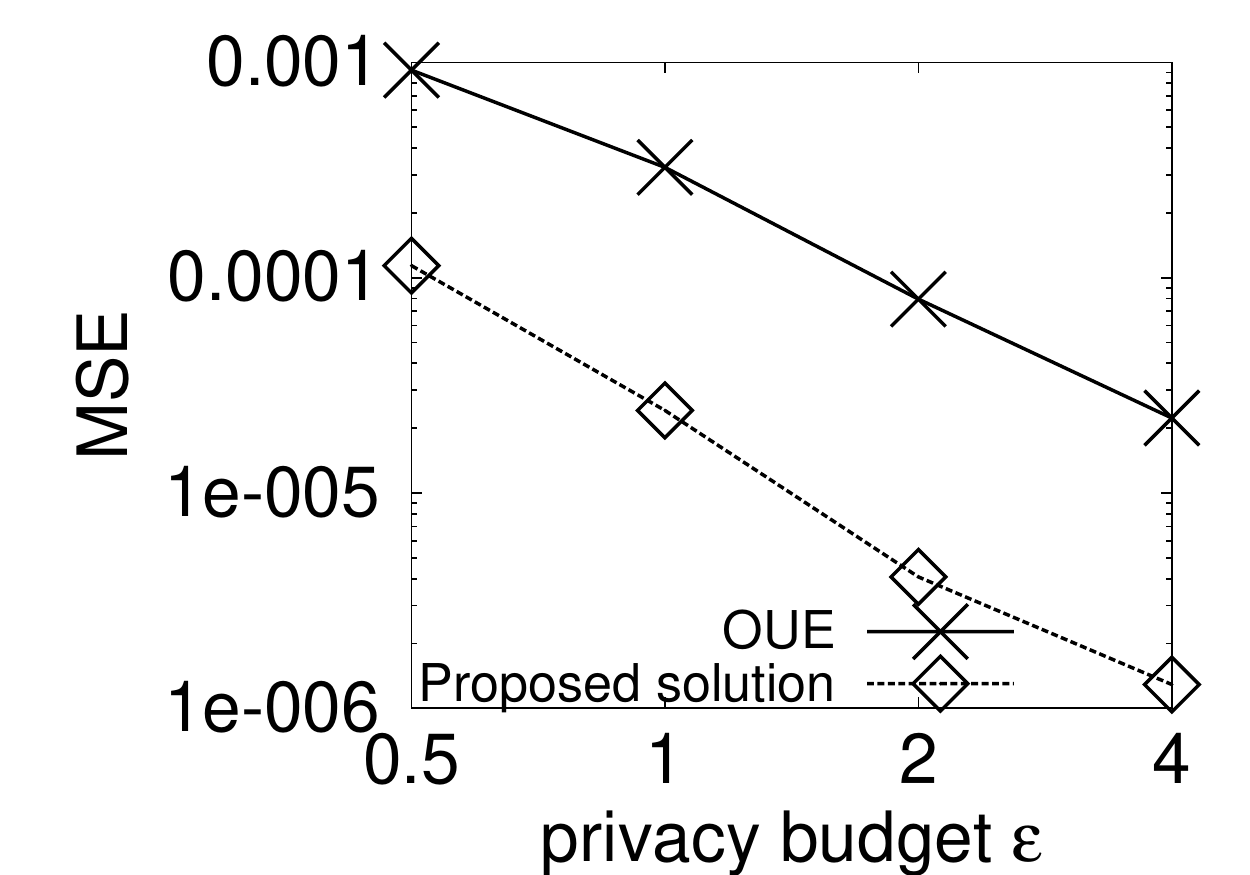} \\
    (a) BR-Numeric & (b) MX-Numeric & (c) BR-Categorical & (d) MX-Categorical
   \end{tabular}\vspace{-5pt}
  \caption{Result accuracy for mean estimation (on numeric attributes) and frequency estimation (on categorical attributes).
  \vspace{-2mm}}
  \label{fig:exp:mean-freq} 
\end{figure*}

\begin{figure*}
  \centering
  \footnotesize
  \begin{tabular}{cccc}
  \multicolumn{4}{c}{\includegraphics[width=0.4\textwidth]{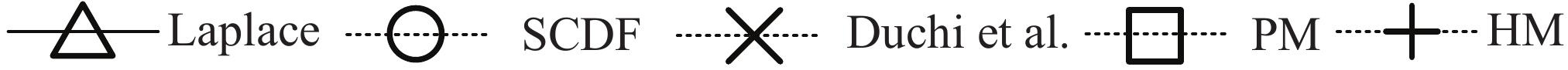}}\\
    \hspace{-3mm}\includegraphics[width=0.23\textwidth]{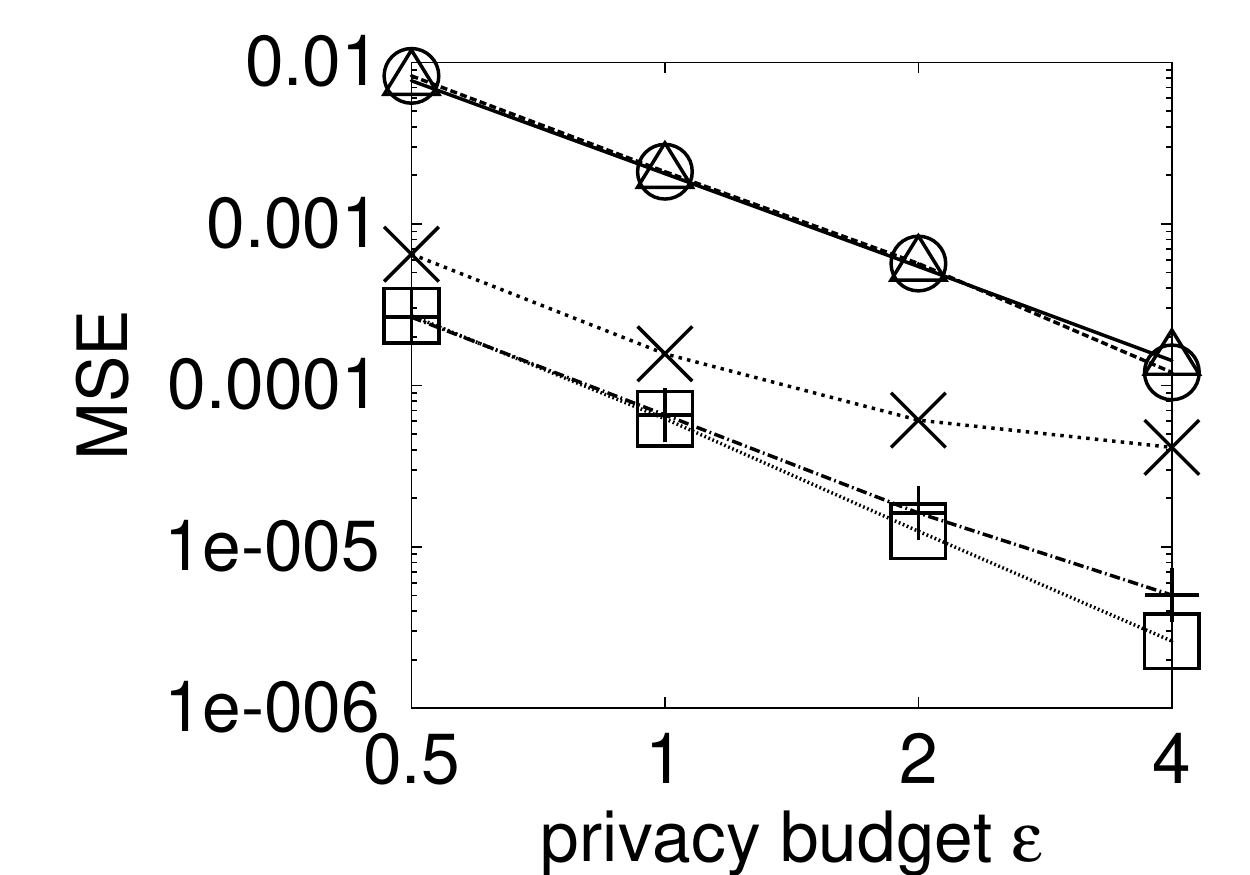} &
    \hspace{-4mm}\includegraphics[width=0.23\textwidth]{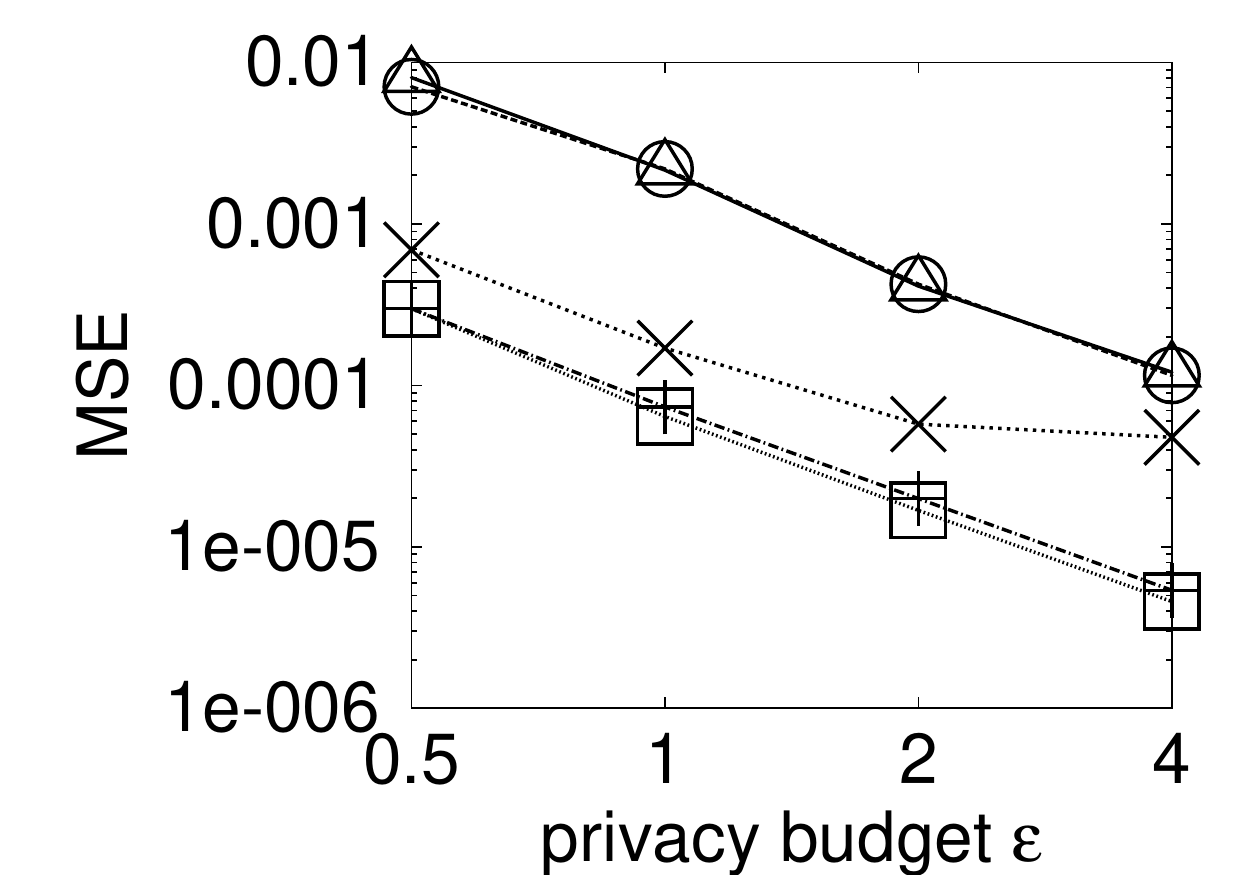} &
    \hspace{-4mm}\includegraphics[width=0.23\textwidth]{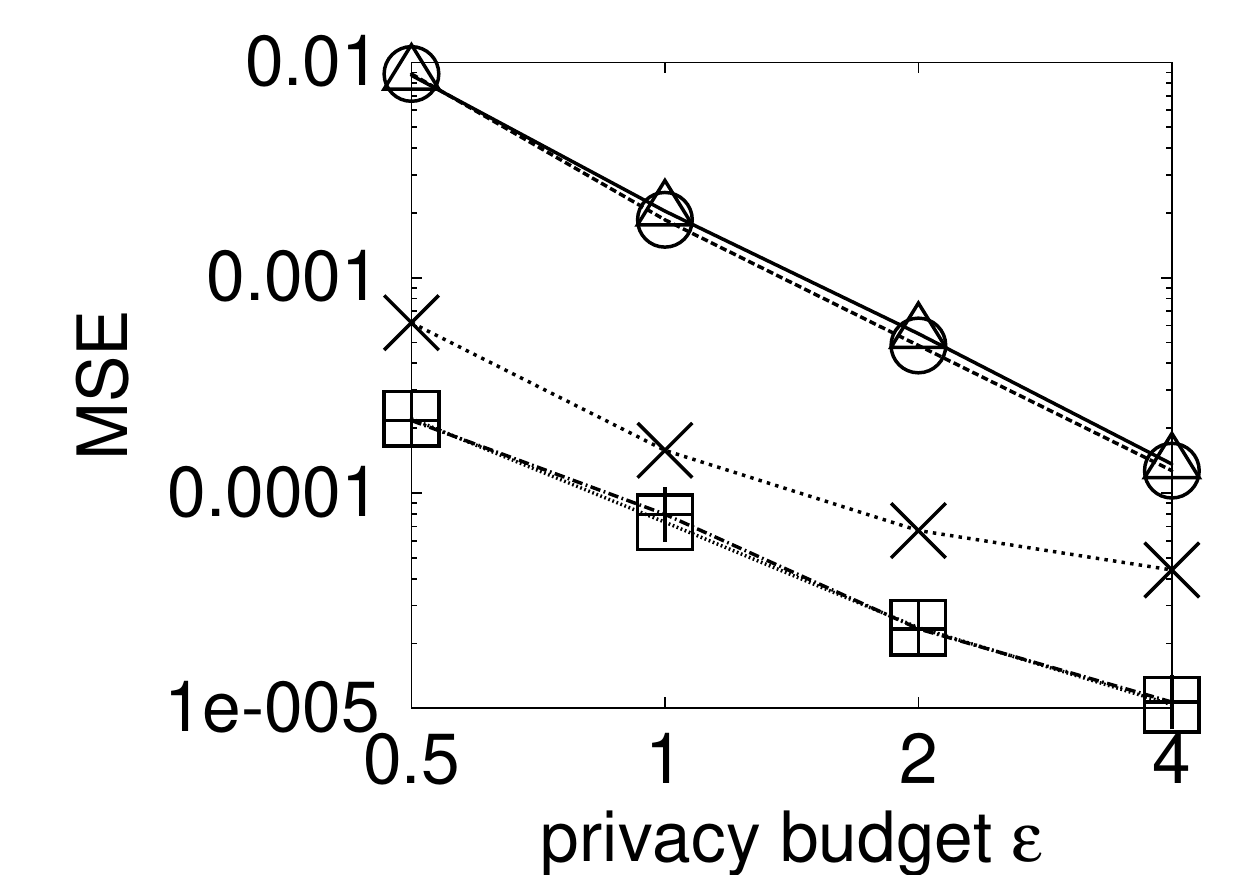} &
    \hspace{-4mm}\includegraphics[width=0.23\textwidth]{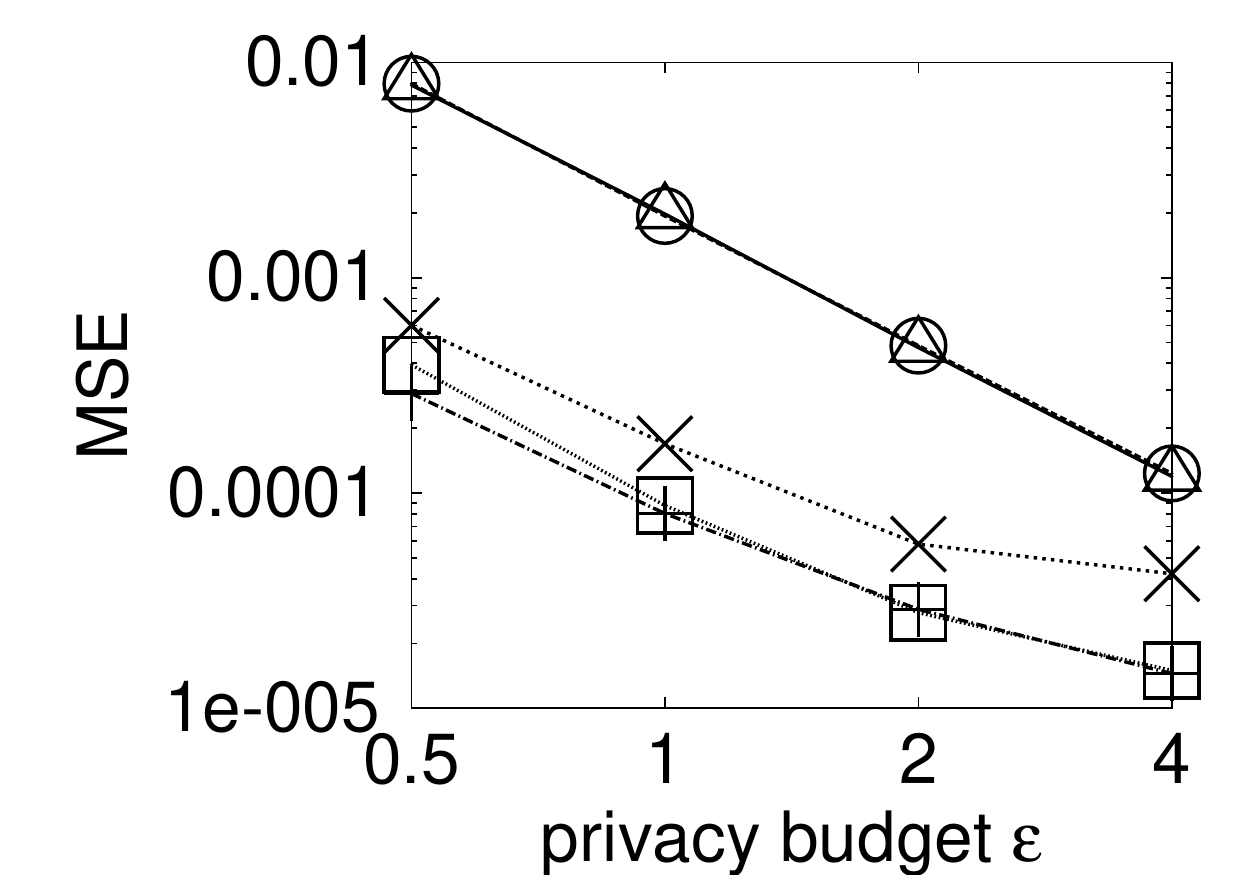} \\
    (a) $\mu = 0$ & (b) $\mu = 1/3$ & (c) $\mu = 2/3$ & (d) $\mu=1$
   \end{tabular}\vspace{-5pt}
  \caption{Result accuracy on synthetic datasets with $16$ dimensions, each of which follows a Gaussian distribution $N(\mu, 1/16)$ truncated to $[-1, 1]$.
  \vspace{-6mm}
  }
  \label{fig:exp:sythetic} 
\end{figure*}

We have implemented the proposed methods and evaluated them using two public datasets extracted from the \textit{Integrated Public Use Microdata Series} \cite{IPUMS}, BR and MX, which contains census records from Brazil and Mexico, respectively. BR contains $4$M tuples and $16$ attributes, among which $6$ are numerical (e.g., age) and $10$ are categorical (e.g., gender); 
MX has $4$M records and $19$ attributes, among which $5$ are numerical and $14$ are categorical. Both datasets contain a numerical attribute ``total income'', which we use as the dependent attribute in linear regression, logistic regression, and SVM (explained further in Section~\ref{sec:exp:erm}). We normalize the domain of each numerical attribute into $[-1, 1]$. In all experiments, we report average results over $100$ runs.

\subsection{Results on Mean Value / Frequency Estimation} \label{sec:exp:mean-freq}

In the first set of experiments, we consider the task of collecting a noisy, multidimensional tuple from each user, in order to estimate the mean of each numerical attribute and the frequency of each categorical value. Since no existing solution can directly support this task, we take the following best-effort approach combining state-of-the-art solutions through the composition property of differential privacy \cite{McSherryT07}. Specifically, let $t$ be a tuple with $d_n$ numeric attributes and $d_c$ categorical attributes. Given total privacy budget $\epsilon$, we allocate $d_n \epsilon/d$ budget to the numeric attributes, and $d_c\epsilon/d$ to the categorical ones, respectively. Then, for the numeric attributes, we estimate the mean value for each of them using either (i) Duchi et al.'s solution (i.e., Algorithm~\ref{alg:duchimulti}), which directly handles multiple numeric attributes, (ii) the Laplace mechanism or (iii) SCDF \cite{Soria-ComasD13}, which is applied to each numeric attribute individually using $\epsilon/d$ budget. The Staircase mechanism leads to similar performance as SCDF, and we omit its results for brevity.
Regarding categorical attributes, since no previous solution addresses the multidimensional case, we apply the optimized unary encoding (OUE) protocol of Wang~et~al.~\cite{lininghui}, the state of the art for frequency estimation on a single categorical attribute, to each attribute independently with $\epsilon/d$ budget. Clearly, by the composition property of differential privacy \cite{McSherryT07}, the above approach satisfies $\epsilon$-LDP.

\begin{figure}
  \centering
  \footnotesize
  \begin{tabular}{cc}
  \multicolumn{2}{c}{\includegraphics[width=0.4\textwidth]{key1.pdf}}\\
    \hspace{-3mm}\includegraphics[width=0.23\textwidth]{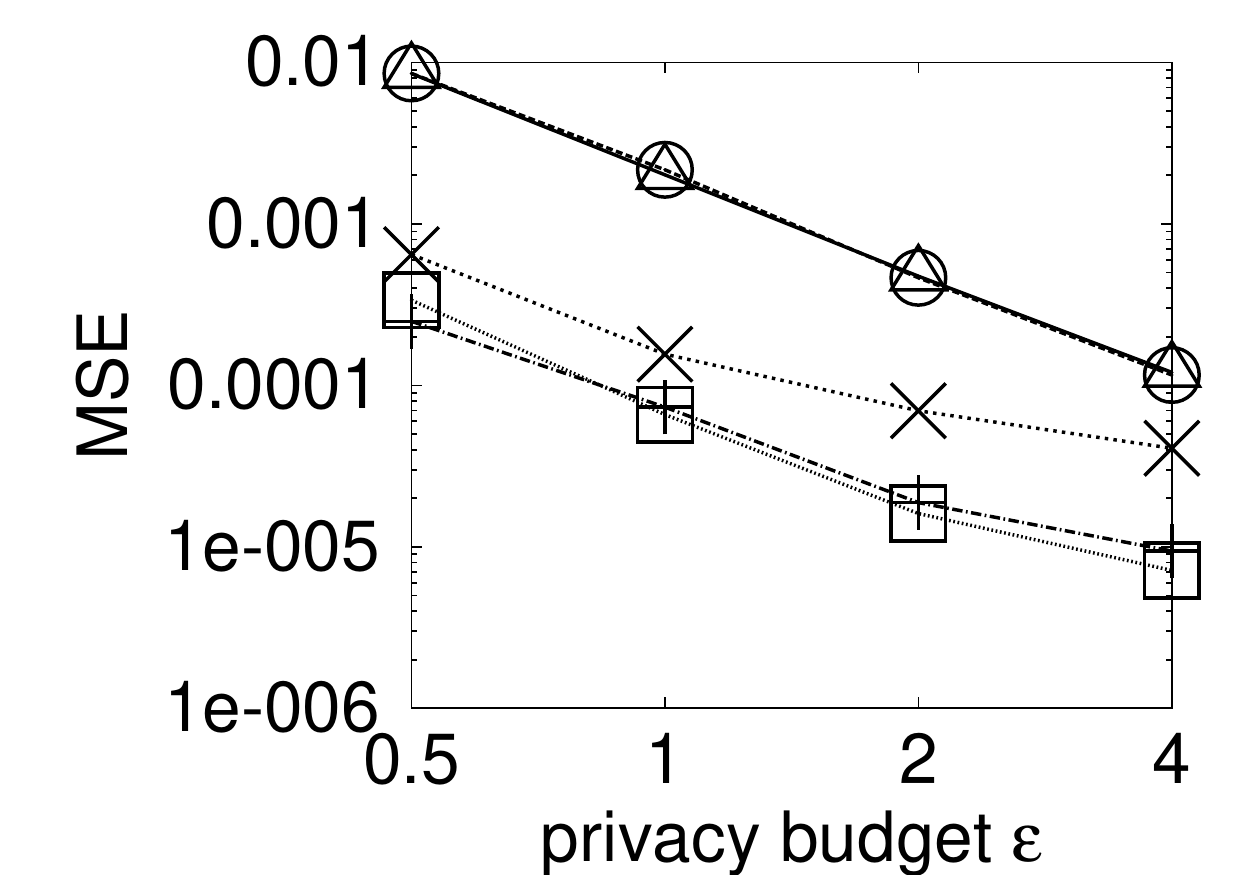} &
    \hspace{-4mm}\includegraphics[width=0.23\textwidth]{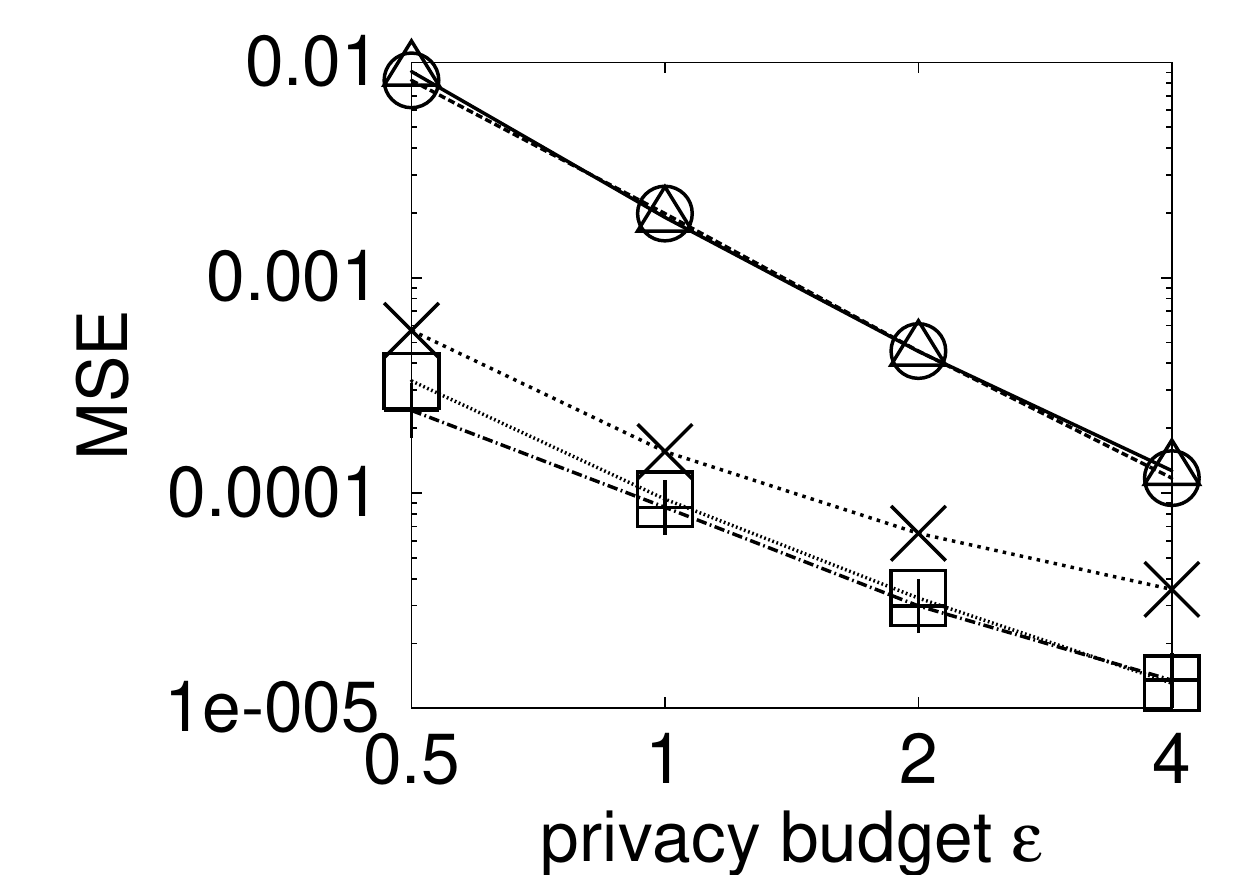} \\
    (a) Uniform distribution & (b) Power law distribution \\
   \end{tabular}\vspace{-2pt}
  \caption{Result accuracy vs. privacy budget on different data distributions. \vspace{-4mm}}
  \label{fig:exp:distribution} 
\end{figure}

We evaluate both the above best-effort approach using existing methods, and the proposed solution in Section~\ref{sec:multi}, on the two real datasets BR and MX. For each method, we measure the mean square error (MSE) in the estimated mean values (for numeric attributes) and value frequencies (for categorical attributes). Fig.~\ref{fig:exp:mean-freq} plots the MSE results as a function of the total privacy budget $\epsilon$. Overall, the proposed solution consistently and significantly outperforms the best-effort approach combining existing methods. One major reason is that the estimation error of the proposed solution is asymptotically optimal, which scales sublinearly to the data dimensionality $d$; in contrast, the best-effort combination of existing approaches involves privacy budget splitting, which is sub-optimal. For instance, on the categorical attributes, applying OUE \cite{lininghui} on each attribute individually leads to $O\left(\frac{d}{\epsilon \sqrt{n}}\right)$ error (where $n$ is the number of users), which grows linearly with data dimensionality $d$. This also explains the consistent performance gap between Duchi~et~al.'s solution \cite{DuchiJW18} and the Laplace mechanism (SCDF mechanism) on numeric attributes. 

Meanwhile, on numeric attributes, the proposed solutions PM and HM outperform Duchi~et~al.'s solution in all settings. This is because (i) although all three methods are asymptotically optimal, Duchi~et~al.'s solution incurs a larger constant than the proposed algorithms and (ii) Duchi~et~al.'s solution cannot handle categorical attributes, and, thus, needs to be combined with OUE through privacy budget splitting, which is sub-optimal. To eliminate the effect of (ii), we run an additional set of experiments with only numeric attributes on several synthetic datasets. Specifically, the first synthetic data contains 16 numeric attributes (i.e., same number of attributes in BR), where each attribute value is generated from a Gaussian distribution with mean value $u$ and standard deviation of $1/4$, but discarding any value that fall out of $[-1, 1]$. 
Fig.~\ref{fig:exp:sythetic} shows the results with varying privacy budget $\epsilon$, and 4 different values for $u$. In all settings, PM and HM outperform Duchi~et~al.'s solution, and the performance gap slightly expands with increasing $\epsilon$, which agrees with our analysis in Section~\ref{sec:multi}. Finally, comparing PM and HM, the difference in their performance is small, and the relative performance of the two can be different in different settings. Note that the main advantage of HM over PM is that on a single numeric attribute, HM is never worse than Duchi et al., whereas PM does not have this guarantee.

We repeat the experiments on two additional synthetic datasets with the same properties as the first synthetic one, except that their attribute values are drawn from different distributions. One follows the uniform distribution where each attribute value is sampled from $[-1, 1]$ uniformly; the other one follows the power law distribution where each attribute value $x$ is sampled from $[-1, 1]$ with probability proportional to $c\cdot (x + 2)^{-10}$. Fig.~\ref{fig:exp:distribution} presents the results, which lead to similar conclusions as the results on real and Gaussian-distributed data.

\begin{figure}
  \centering
  \footnotesize
  \begin{tabular}{cc}
  \multicolumn{2}{c}{}\\
    \hspace{-3mm}\includegraphics[width=0.23\textwidth]{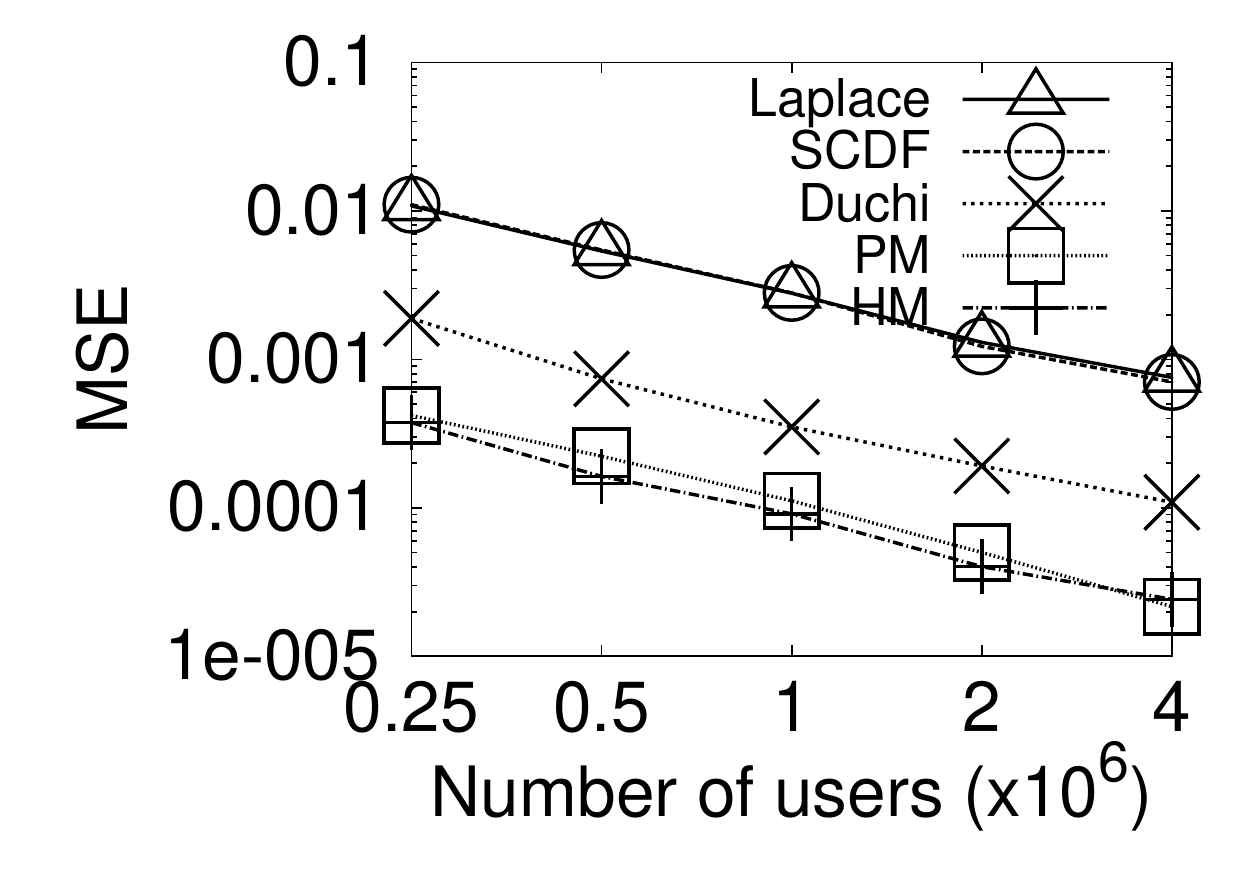} &
    \hspace{-4mm}\includegraphics[width=0.23\textwidth]{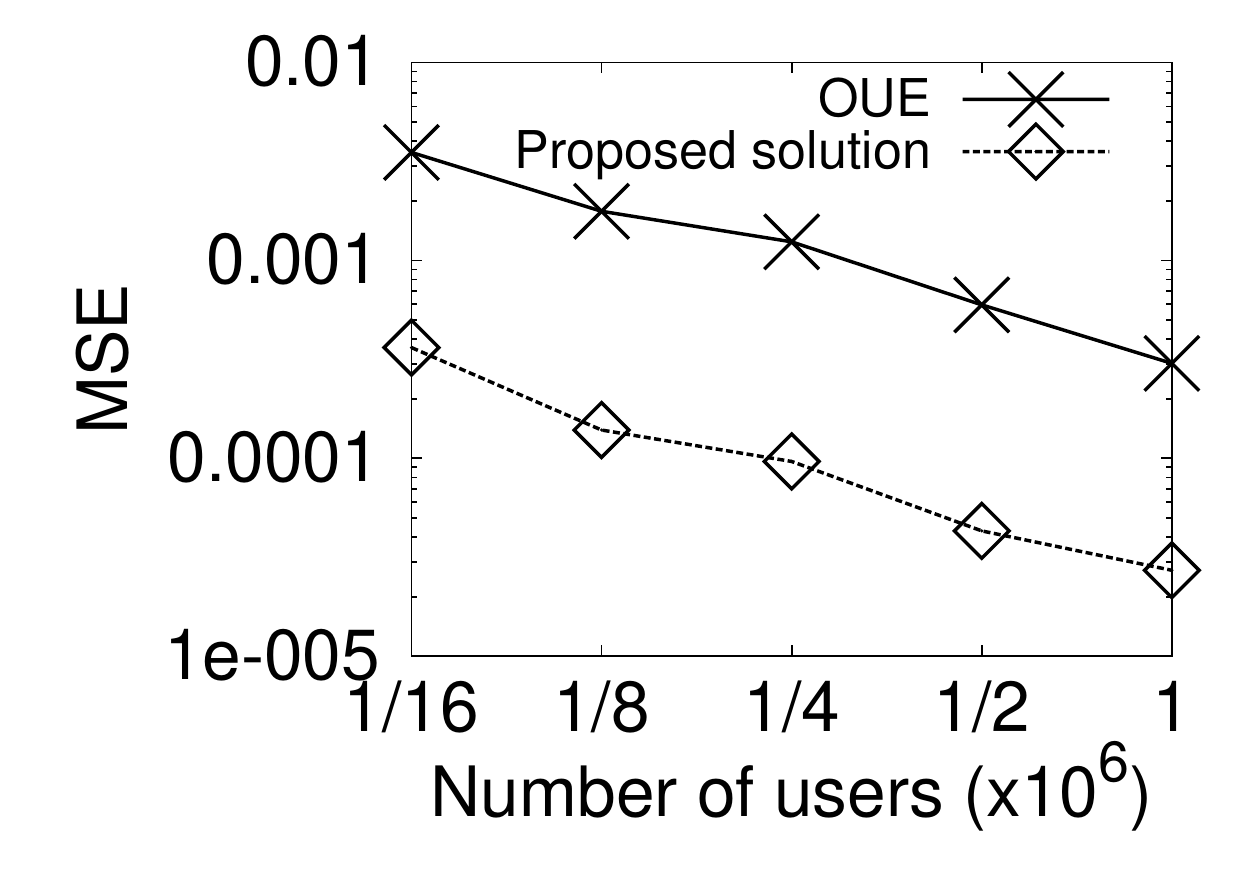} \\
    (a) Numeric & (b) Categorical \\
   \end{tabular}\vspace{-2pt}
  \caption{Result accuracy vs. number of users. \vspace{-4mm}}
  \label{fig:exp:varyn} 
\end{figure}

\begin{figure}
  \centering
  \footnotesize
  \begin{tabular}{cc}
  \multicolumn{2}{c}{}\\
    \hspace{-3mm}\includegraphics[width=0.23\textwidth]{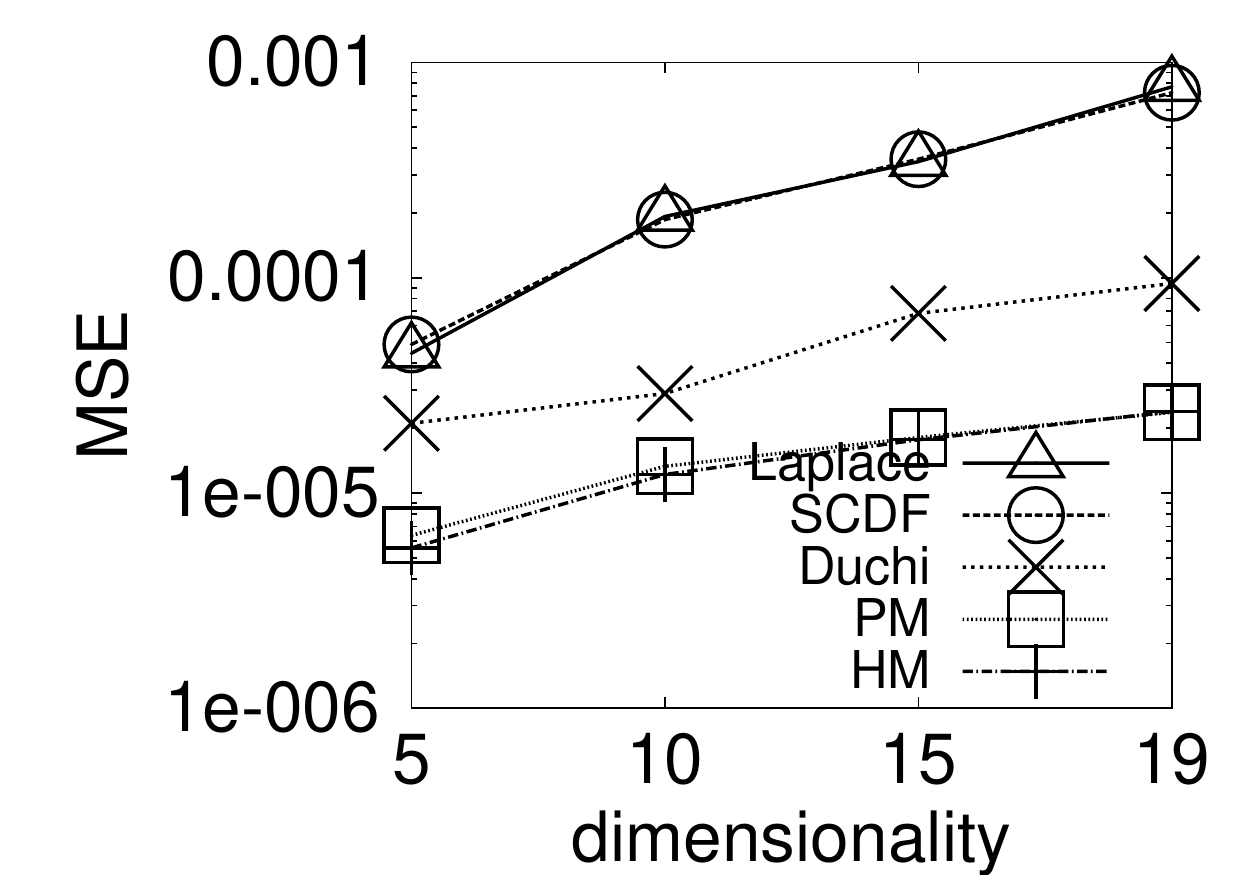} &
    \hspace{-4mm}\includegraphics[width=0.23\textwidth]{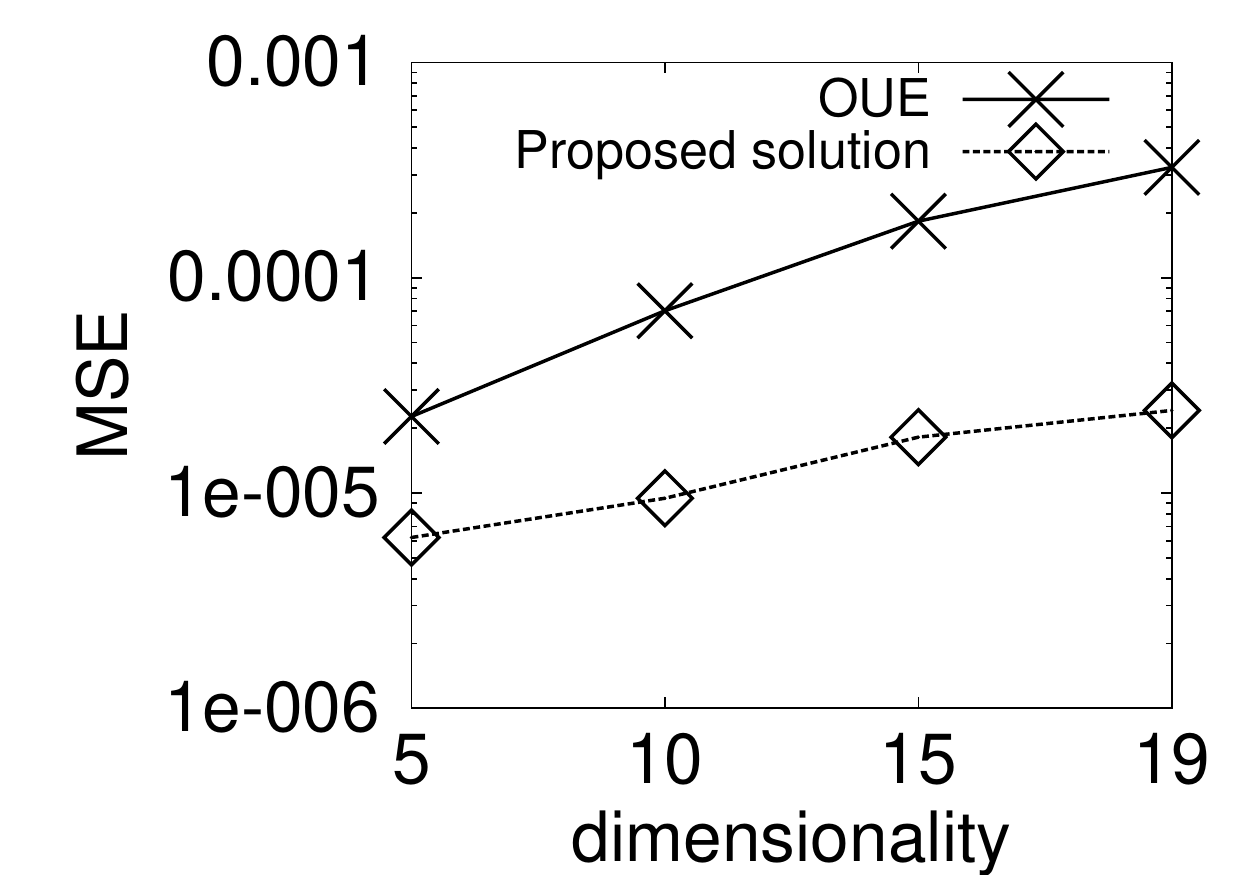} \\
    (a) Numeric & (b) Categorical \\
   \end{tabular}\vspace{-2pt}
  \caption{Result accuracy vs. dimensionality. \vspace{-4mm}}
  \label{fig:exp:varyd} 
\end{figure}

Lastly, Figs.~\ref{fig:exp:varyn} and ~\ref{fig:exp:varyd} show the result accuracy in terms of MSE with varying the number of users $n$ and dimensionality $d$ on the MX dataset. Observe that more users and lower dimensionality both lead to more accurate results, which agrees with the theoretical analysis in Lemma \ref{lmm:multi-accuracy}. Meanwhile, in all settings the proposed solutions consistently outperform their competitors by clear margins. In the next subsection, we omit the results for SCDF, which are comparable to that of the Laplace mechanism.

\subsection{Results on Empirical Risk Minimization}\label{sec:exp:erm}

In the second set of experiments, we evaluate the accuracy performance of the proposed methods for linear regression, logistic regression, and SVM classification on BR and MX. For both datasets, we use the numeric attribute ``total income'' as the dependent variable, and all other attributes as independent variables. Following common practice, we transform each categorical attribute $A_j$ with $k$ values into $k-1$ binary attributes with a domain $\{0, 1\}$, such that (i) the $l$-th ($l < k$) value in $A_j$ is represented by $1$ on the $l$-th binary attribute and $0$ on each of the remaining $k-2$ attributes, and (ii) the $k$-th value in $A_j$ is represented by $0$ on all binary attributes. After this transformation, the dimensionality of BR (resp.\ MX) becomes $90$ (resp.\ $94$). 
For logistic regression and SVM, we also covert ``total income'' into a binary attribute by mapping the values larger than the mean value to 1, and 0 otherwise.

Since each user sends gradients to the aggregator, which are all numeric, the experiment involves the 4 competitors in Section \ref{sec:exp:mean-freq} for numeric data: PM, HM, Duchi~et~al.~\cite{DuchiJW18}, and the Laplace mechanism applied to each attribute independently with equally split privacy budget (i.e., $\epsilon/d$ for each attribute). Additionally, we also include the result in the non-private setting. For all methods, we set the regularization factor $\lambda = 10^{-4}$. On each dataset, we use $10$-fold cross validation 5 times to assess the performance of each method. 

Fig.~\ref{fig:exp:lr} and Fig.~\ref{fig:exp:svm} show the misclassification rate of each method for logistic regression and SVM classification, respectively, with varying values of the privacy budget $\epsilon$. Similar to the results in Section \ref{sec:exp:mean-freq}, the Laplace mechanism leads to significantly higher than the other three solutions, due to the fact that its error rate is sub-optimal. The proposed algorithms PM and HM consistently outperform Duchi~et~al.'s solution with clear margins, since (i) the former two have smaller constant as analyzed in Section \ref{sec:multi}, and (ii) the gradient of each user often consists of elements whose absolute values are small, for which PM and HM are particularly effective, as we mention in Section~\ref{sec:PM}. Further, in some settings such as SVM with $\epsilon \geq 2$ on BR, the accuracy of PM and HM approaches that of the non-private method. Comparing the results with those in Section \ref{sec:exp:mean-freq}, we observe that the misclassification rates for logistic regression and SVM classification do not drop as quickly with increasing privacy budget $\epsilon$ as in the case of MSE for mean values and frequency estimates. This is due to the inherent stochastic nature of SGD: that accuracy in gradients does not have a direct effect on the accuracy of the model. For the same reason, there is no clear trend for the performance gap between PM/HM and Duchi~et~al.'s solution.

Fig.~\ref{fig:exp:linearr} demonstrates the mean squared error (MSE) of the linear regression model generated by each method with varying $\epsilon$. We omit the MSE results for the Laplace mechanism, since they are far higher than the other three methods. The proposed solutions PM and HM once again consistently outperform Duchi~et~al.'s solution. Overall, our experimental results demonstrate the effectiveness of PM and HM for empirical risk minimization under local differential privacy, and their consistent performance advantage over existing approaches.

\begin{figure}
  \centering
  \footnotesize
  \begin{tabular}{cc}
  \multicolumn{2}{c}{\includegraphics[width=0.4\textwidth]{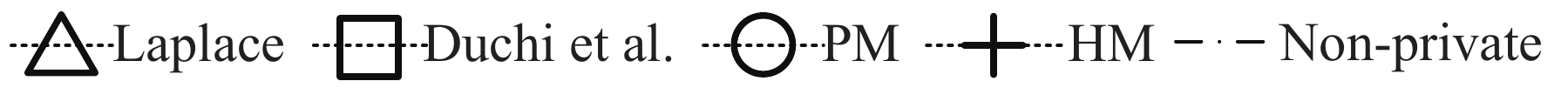}}\\
    \hspace{-3mm}\includegraphics[width=0.23\textwidth]{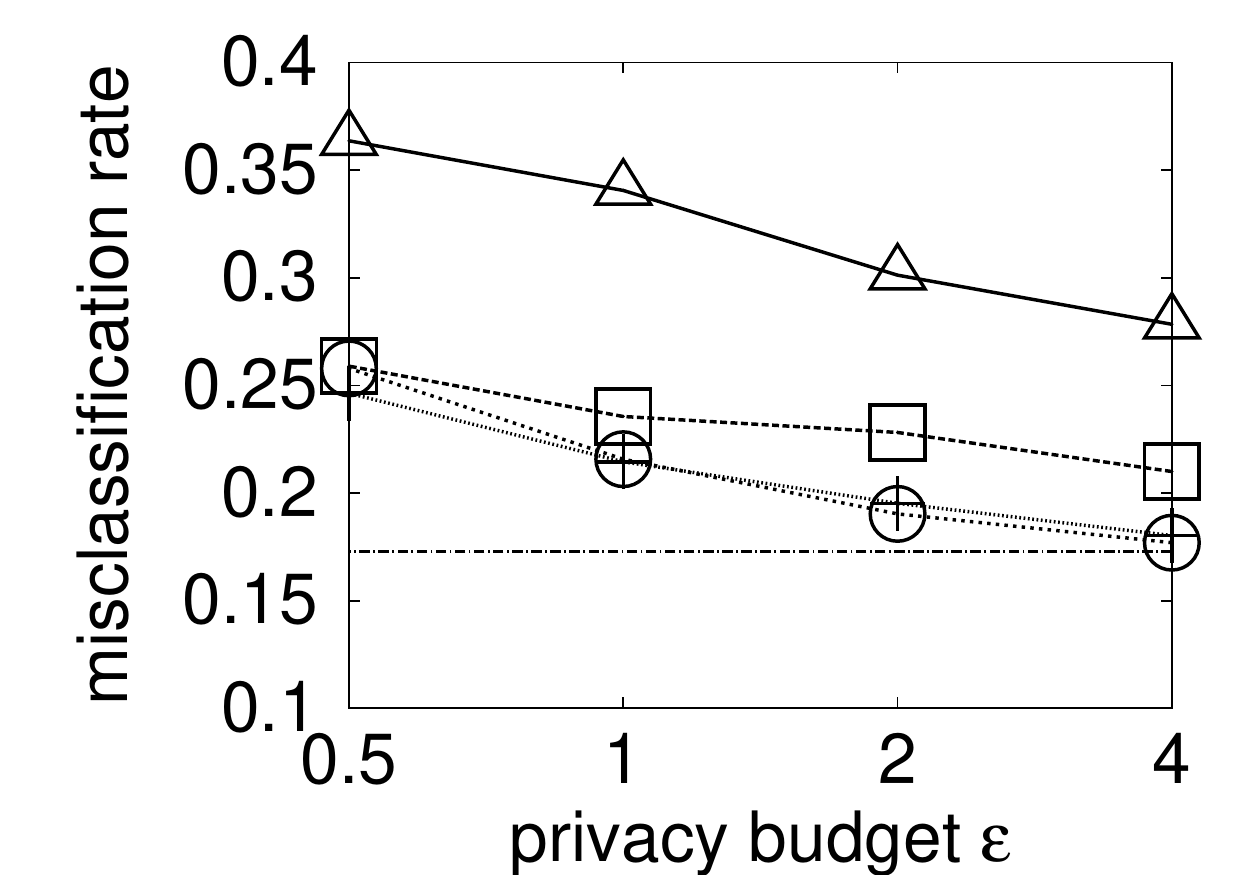} &
    \hspace{-4mm}\includegraphics[width=0.23\textwidth]{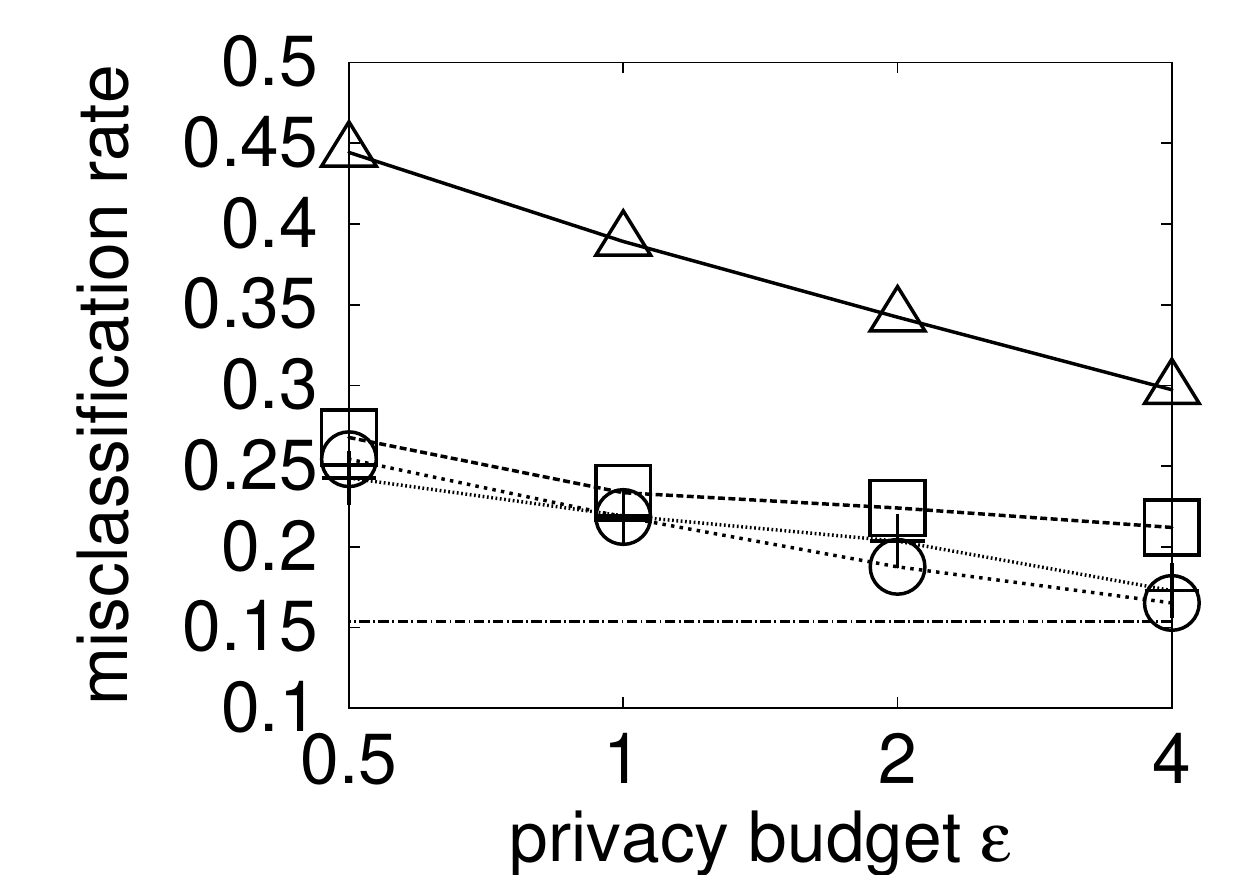} \\
    (a) BR & (b) MX \\

   \end{tabular}\vspace{-2pt}
  \caption{Logistic Regression.  \vspace{-4mm}}
  \label{fig:exp:lr} 
\end{figure}

\begin{figure}
  \centering
  \footnotesize
  \begin{tabular}{cc}
  \multicolumn{2}{c}{\includegraphics[width=0.4\textwidth]{key3.pdf}}\\
    \hspace{-3mm}\includegraphics[width=0.23\textwidth]{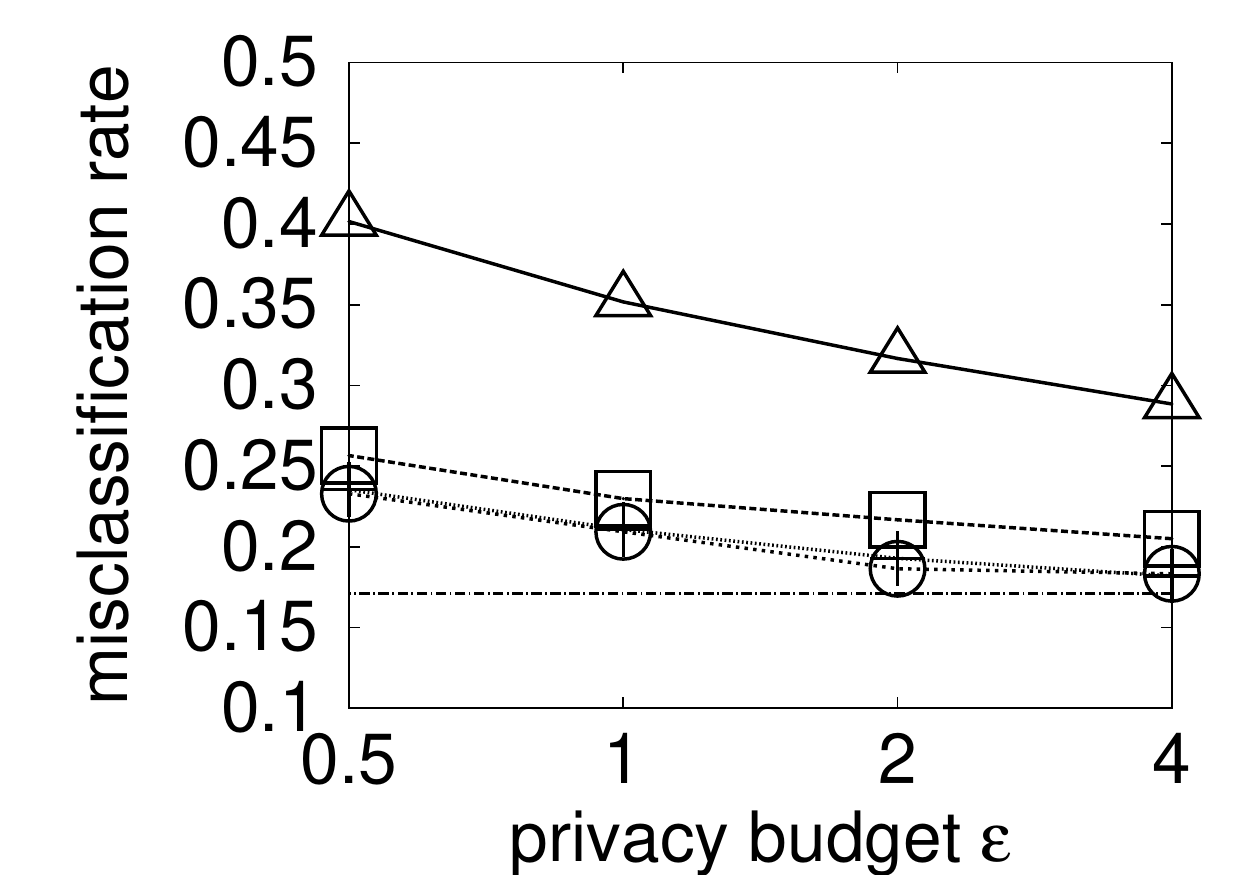} &
    \hspace{-4mm}\includegraphics[width=0.23\textwidth]{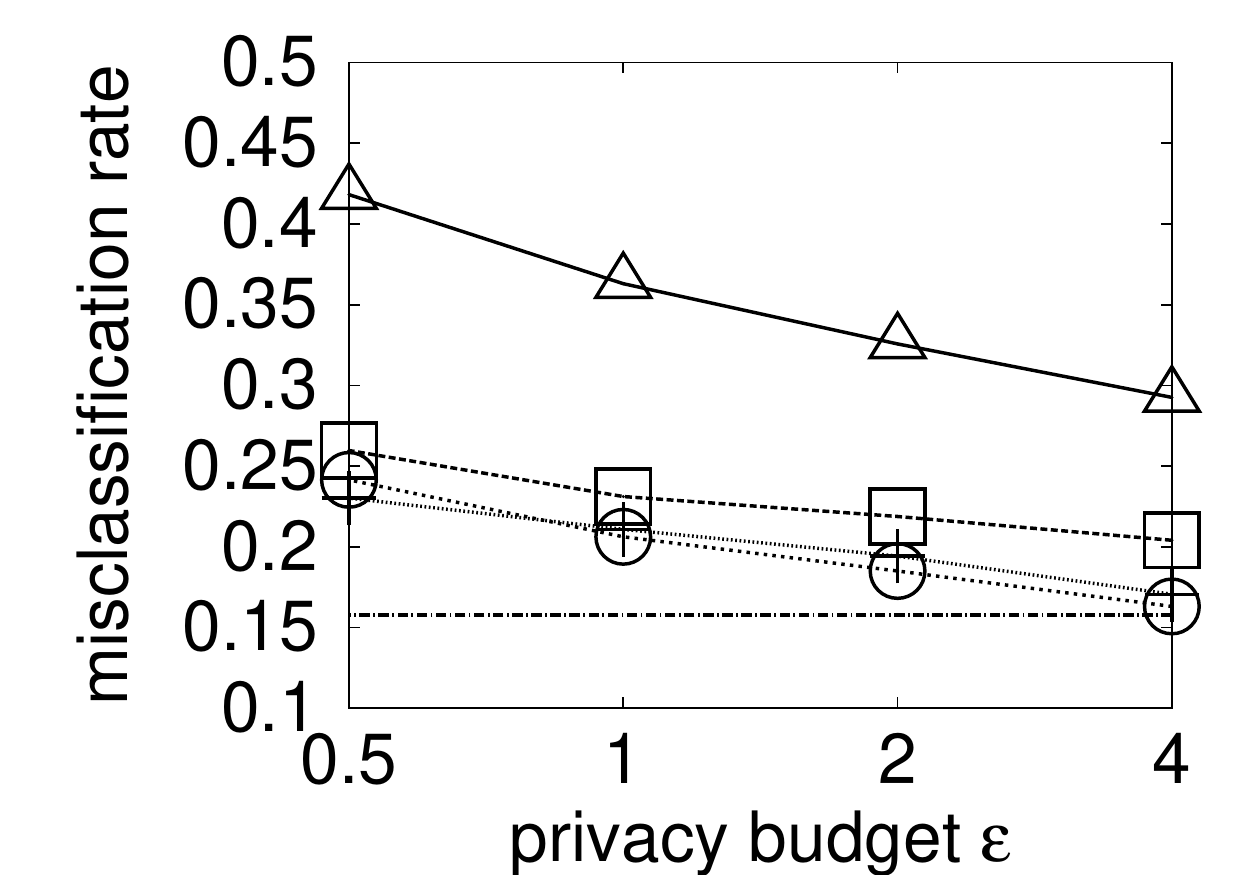} \\
    (a) BR & (b) MX \\
   \end{tabular}\vspace{-2pt}
  \caption{Support Vector Machines
(SVM). \vspace{-4mm}}
  \label{fig:exp:svm} 
\end{figure}

\begin{figure}
  \centering
  \footnotesize
  \begin{tabular}{cc}
  \multicolumn{2}{c}{\includegraphics[width=0.4\textwidth]{key3.pdf}}\\
    \hspace{-3mm}\includegraphics[width=0.23\textwidth]{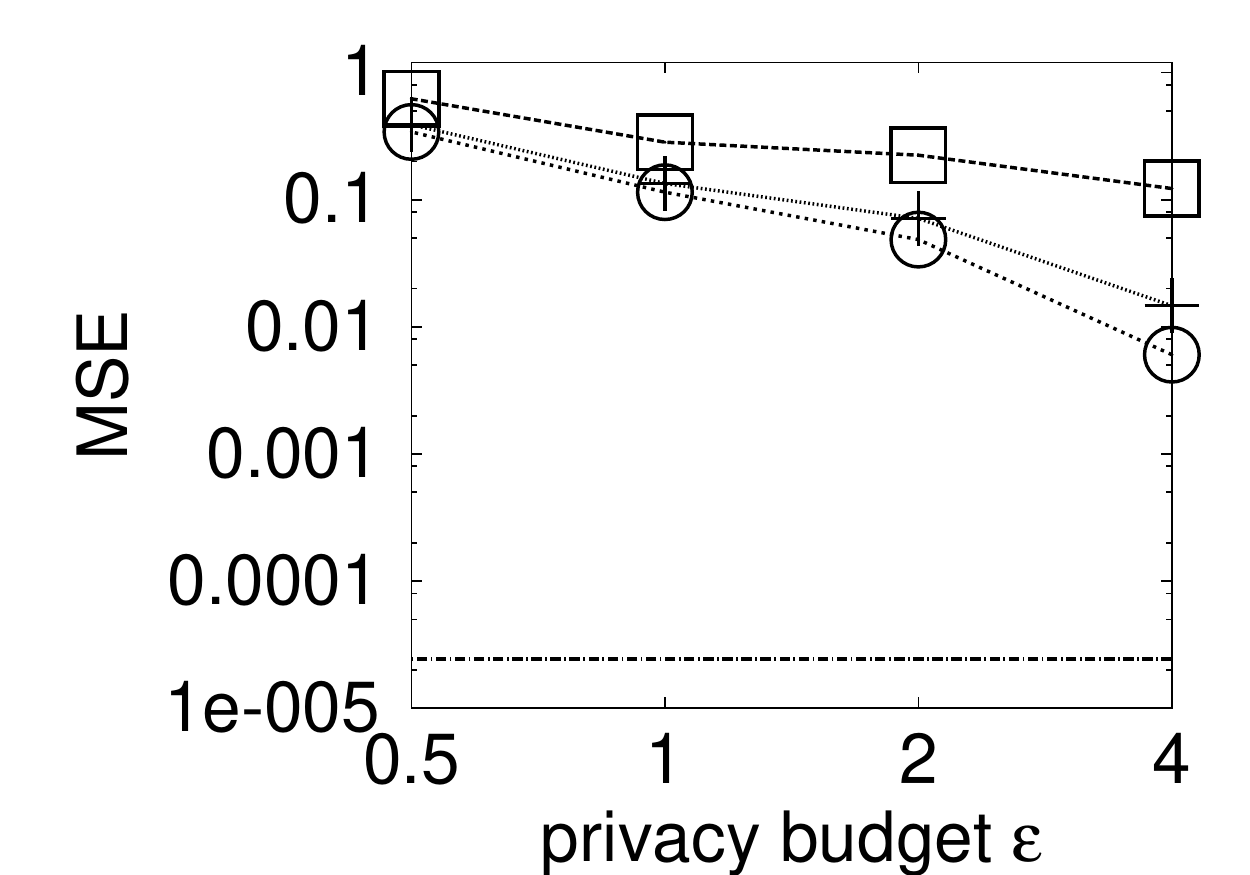} &
    \hspace{-4mm}\includegraphics[width=0.23\textwidth]{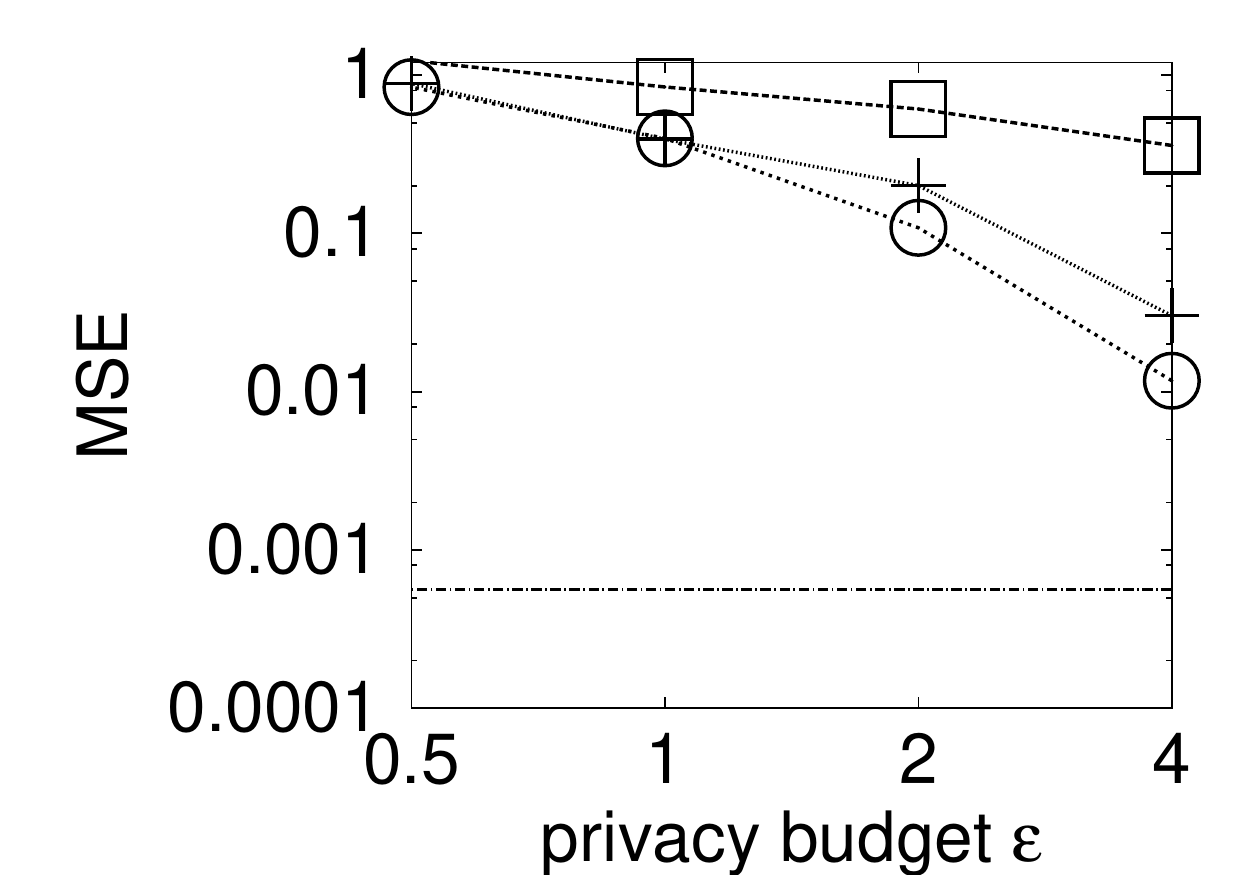} \\
    (a) BR & (b) MX \\
   \end{tabular}\vspace{-2pt}
  \caption{Linear Regression. \vspace{-4mm}}
  \label{fig:exp:linearr} 
\end{figure}
\section{Related Work} \label{sec:related}

Differential privacy \cite{DworkLaplace} is a strong privacy standard that provides semantic, information-theoretic guarantees on individuals' privacy, which has attracted much attention from various fields,
including data management \cite{ChenLQKJ16,cormode2018marginal,krishnan2016privateclean}, machine learning \cite{bassily2017practical},
theory \cite{BS15,duchi2013local,kasiviswanathan2008can},
and systems \cite{bittau2017prochlo}. Earlier models of differential privacy \cite{DworkLaplace,DworkR14,McSherryT07} rely on a trusted data curator, who collects and manages the exact private information of individuals, and releases statistics derived from the data under differential privacy requirements.
Recently, 
much attention has been shifted to the local differential privacy (LDP) model (e.g., \cite{kasiviswanathan2008can, duchi2013local}), which eliminates the data curator and the collection of exact private information.

LDP can be connected to the classical randomized response technique in surveys~\cite{Warner65}. 
Erlingsson~et~al.~\cite{RAPPOR2014} propose the RAPPOR framework, which is based on the randomized response mechanism for publishing a value for binary attributes under LDP. They use this mechanism with a Bloom filter, which intuitively adds another level of protection and increases the difficulty for the adversary to infer private information.
A follow-up paper \cite{fanti2016building} extends RAPPOR to more complex statistics such as joint-distributions and association testing, as well as categorical attributes that contain a large number of potential values, such as a user's home page. 
Wang~et~al.~\cite{lininghui} investigate the same problem, and propose a different method: they transform $k$ possible values into a noisy vector with $k$ elements, and  send the latter to curator. Bassily and Smith~\cite{BS15} propose an asymptotically optimal solution for building succinct histograms over a large categorical domain under LDP. 
Note that all of the above methods focus on a single categorical attribute, and, thus, are orthogonal to our work on multidimensional data including numeric attributes.
Ren~et~al.~\cite{ren2018lopub} investigate the problem of publishing multiple attributes, and employ the idea of $k$-sized vector, similar to \cite{lininghui}. This approach, however, incurs rather high communication costs between the aggregator and the users, since it involves the transmission of multiple $k$-sized vectors.
Duchi~et~al.~\cite{duchi2013local}
propose the minimax framework for LDP based on information theory, prove upper and lower error bounds of LDP-compliant methods, and analyze the trade-off between privacy and accuracy. Besides,
Kairouz~et~al.~\cite{kairouz2014extremal} propose the extremal mechanisms, which are a family of LDP mechanisms for data with discrete inputs, i.e., each input domain $\mathcal{X}$ contains a finite number of possible values. These mechanisms have an output distribution $pdf$ with a key property: for any input $x \in \mathcal{X}$ and any output $y$, $Pr[y \mid x]$ has only two possible values that differ by a factor of $\exp(\epsilon)$. Kairouz et al.\ show that for any given utility measure, there exists an extremal mechanism with optimal utility under this measure, using a linear program with $2^{|\mathcal{X}|}$ variables. It is unclear how to apply extremal mechanisms to continuous input domains with an infinite number of possible values, which is the focus on this paper.


Various data analytics and machine learning problems have been studied under LDP, such as probability distribution estimation~\cite{duchi2013local2,kairouz2016discrete,pastore2016locally,ye2017optimal,murakami2018toward},
heavy hitter discovery~\cite{bassily2017practical,bun2018heavy,QinYYKXR16,wang2017locally}, frequent new term discovery~\cite{NingWang2018}, frequency estimation~\cite{lininghui,BS15},
 frequent itemset mining~\cite{wang2018locally}, marginal release~\cite{cormode2018marginal},  clustering~\cite{NissimStemmer18}, 
 and
hypothesis testing~\cite{gaboardi2017local}.

Finally, a recent work~\cite{avent2017blender}  introduces a hybrid model that involves both centralized and local differential privacy.
Bittau~et~al.~\cite{bittau2017prochlo} evaluate real-world implementations of LDP. Also, LDP has been considered in several applications including the collection of indoor positioning data~\cite{kim2018application}, inference control on mobile
sensing~\cite{liu2017deeprotect}, and the publication of crowdsourced data~\cite{ren2018lopub}.



\section{Conclusion}\label{sec:conclusion}

This work systematically investigates the problem of collecting and analyzing users' personal data under $\epsilon$-local differential privacy, in which the aggregator only collects randomized data from the users, and computes statistics based on such data. The proposed solution is able to collect data records that contain multiple numerical and categorical attributes, and compute accurate statistics from simple ones such as mean and frequency to complex machine learning models such as linear regression, logistic regression and SVM classification. Our solution achieves both optimal asymptotic error bound and high accuracy in practice. 
In the next step, we plan to apply the proposed solution to more complex data analysis tasks such as deep neural networks. 





 \section*{Acknowledgment}

This research was supported by National Natural Science Foundation of China (61672475, 61433008), by Samsung via a GRO grant, by the National Research Foundation, Prime Minister's Office, Singapore under its Strategic Capability Research Centres Funding Initiative, by Qatar National Research Fund (NPRP10-0208-170408), and by Nanyang Technological University Startup Grant (M4082311.020). We thank T.~Nguyen for his contributions to an earlier version of this work.

\balance

\end{document}